\documentclass[11pt,a4paper]{article}
\usepackage[lmargin=1.0in,rmargin=1.0in,bottom=1.0in,top=1.0in,twoside=False]{geometry}

\usepackage{inconsolata}
\usepackage{libertine}

\usepackage{amssymb,amsmath,amsthm}
\usepackage{mathtools}
\usepackage{graphicx}
\usepackage{enumerate}
\usepackage{tikz}
\usetikzlibrary{shapes}

\usepackage[T1]{fontenc}

\usepackage{enumitem}

\usepackage{microtype}
\usepackage{amsfonts}
\usepackage{comment}
\usepackage{mathrsfs}
\usepackage{todonotes}

\definecolor{blue}{rgb}{0.1,0.2,0.5}
\definecolor{brown}{rgb}{0.6,0.6,0.2}
\usepackage[ocgcolorlinks, linkcolor={blue}, citecolor={brown}]{hyperref}
\usepackage{cleveref}

\usepackage{enumerate}
\usepackage{latexsym}

\usepackage{soul}

\newtheorem{lemma}{Lemma}[section]
\newtheorem{observation}{Observation}[section]
\newtheorem{claim}{Claim}[section]
\newtheorem{theorem}[lemma]{Theorem}
\newtheorem{corollary}[lemma]{Corollary}
\newtheorem{conjecture}[lemma]{Conjecture}

\newcommand{\Oh}{\mathcal{O}}

\newcommand{\argmax}{\text{argmax}}
\newcommand{\argmin}{\text{argmin}}

\newcommand{\distprofile}[3]{\mathrm{prof}_{#1,#2}[#3]}

\newcommand{\hatprofile}[3]{\widehat{\mathrm{prof}}_{#1,#2}[#3]}
\newcommand{\diststarprofile}[3]{\mathrm{prof}^\star_{#1,#2}[#3]}
\newcommand{\dist}{\mathrm{dist}}
\newcommand{\ecc}{\mathrm{ecc}}

\renewcommand{\leq}{\leqslant}
\renewcommand{\geq}{\geqslant}
\renewcommand{\le}{\leqslant}
\renewcommand{\ge}{\geqslant}
\renewcommand{\setminus}{-}

\usepackage{marvosym}

\newcommand{\cutgraph}{\textrm{\CutLeft}}

\usepackage{textpos}

\begin{document}

\author{
  Kacper Kluk%
  \thanks{Institute of Informatics, University of Warsaw, Poland. \texttt{k.kluk@uw.edu.pl}}
  \and
  Marcin Pilipczuk%
  \thanks{Institute of Informatics, University of Warsaw, Poland. \texttt{m.pilipczuk@uw.edu.pl}}
  \and
  Micha\l{} Pilipczuk%
  \thanks{Institute of Informatics, University of Warsaw, Poland. \texttt{michal.pilipczuk@mimuw.edu.pl}}
  \and
  Giannos Stamoulis%
  \thanks{IRIF, Université Paris Cité, CNRS, Paris, France. \texttt{giannos.stamoulis@irif.fr}}
}

\title{Faster diameter computation in graphs of bounded Euler genus%
  \thanks{
  K.K. and Ma.P. are supported by Polish National Science Centre SONATA BIS-12 grant number 2022/46/E/ST6/00143.
  This work is a part of project BOBR (Mi.P., G.S.) that has received funding from the European Research Council (ERC) 
under the European Union's Horizon 2020 research and innovation programme (grant agreement No.~948057). In particular, a majority of work on this manuscript was done while G.S. was affiliated with University of Warsaw.}}

\date{}

\maketitle

\begin{abstract}
We show that for any fixed integer $k \geq 0$, there exists an algorithm
that computes the diameter and the eccentricies of all vertices of an input
unweighted, undirected $n$-vertex graph of Euler genus at most $k$ in time
 \[ \Oh_k(n^{2-\frac{1}{25}}). \]
Furthermore, for the more general class of graphs
that can be constructed
by clique-sums from graphs that are of Euler genus at most $k$ after deletion
of at most $k$ vertices, we show an algorithm for the same task that achieves the running time bound
 \[ \Oh_k(n^{2-\frac{1}{356}} \log^{6k} n). \]
Up to today, the only known subquadratic algorithms for computing the diameter in those graph classes
are that of [Ducoffe, Habib, Viennot; SICOMP 2022], [Le, Wulff-Nilsen; SODA 2024],
and [Duraj, Konieczny, Pot\k{e}pa; ESA 2024]. These algorithms work
in the more general setting of $K_h$-minor-free graphs,
but the running time bound is $\Oh_h(n^{2-c_h})$ for some constant $c_h > 0$ depending
on $h$. 
That is, our savings in the exponent, as compared to the naive quadratic algorithm, 
are independent of the parameter $k$.

The main technical ingredient of our work is an improved bound on the number of distance
profiles, as defined in [Le, Wulff-Nilsen; SODA 2024], in graphs of bounded Euler genus.
\end{abstract}

\begin{textblock}{20}(11.8,4.35)
  \includegraphics[width=60px]{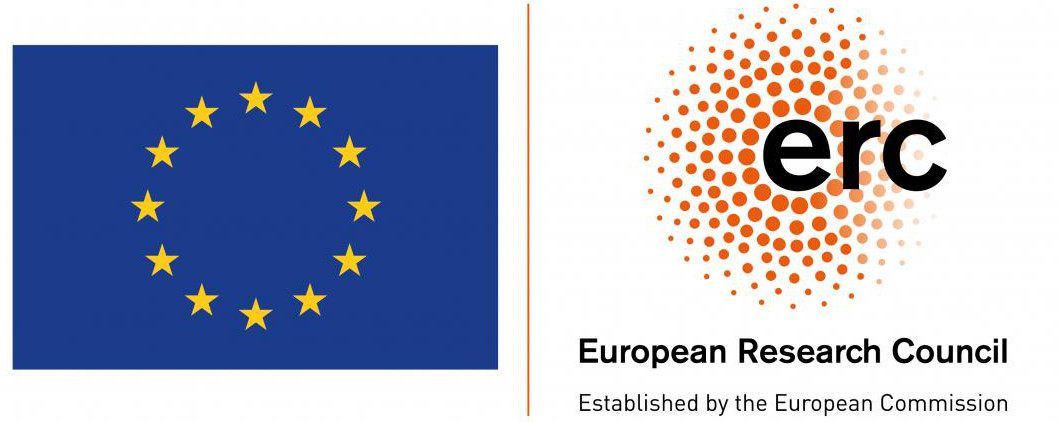}
\end{textblock}

\section{Introduction}

Computing the diameter of an input (undirected, unweighted) graph $G$ is a classic computational 
problem that can be trivially solved in $\Oh(nm)$ time\footnote{We follow the convention that the vertex and the edge count of the input graph are denoted by $n$ and $m$, respectively.}.
In 2013, Roditty and Vassilevska-Williams showed that this running
time bound cannot be significantly improved in general:
any algorithm distinguishing graphs of diameter $2$ and $3$
running in time $\Oh(m^{2-\varepsilon})$, for any fixed $\varepsilon > 0$,
would break the Strong Exponential Time Hypothesis~\cite{RodittyW13}. 
This motivates the search for restrictions on $G$ that would make the  problem of computing the diameter more tractable.

As shown by Cabello and Knauer~\cite{CabelloK09}, sophisticated orthogonal range query data structures allow near-linear diameter
computation in graphs of constant treewidth.
A breakthrough result by Cabello~\cite{Cabello19} showed that 
the diameter of an $n$-vertex planar graph can be computed in $\widetilde{\Oh}(n^{11/6})$ time;
this complexity has been later improved by Gawrychowski, Kaplan, Mozes, Sharir, and Weimann to
$\widetilde{\Oh}(n^{5/3})$~\cite{GawrychowskiKMS21}\footnote{The $\widetilde{\Oh}(\cdot)$ notation hides factors polylogarithmic in $n$, and the $\Oh_k(\cdot)$ notation hides factors depending on a parameter $k$.}.
A subsequent line of research~\cite{DucoffeHV22,DurajKP23,LeW24}
generalized this result to $K_h$-minor-free graphs:
for every integer $h$, there exists a constant $c_h > 0$ such that the diameter
problem in $n$-vertex $K_h$-minor-free graphs can be solved in time $\Oh_h(n^{2-c_h})$. 
In the works~\cite{DurajKP23,LeW24}, it holds that $c_h = \Omega\left(\frac{1}{h}\right)$; so the savings tend to zero as the size of the excluded clique minor increases.

However, known lower bounds, including the one of~\cite{RodittyW13}, does not exclude the possibility
that $c_h$ can be made a universal constant. That is, no known lower bound refutes
the following conjecture:
\begin{conjecture}\label{conj:taunt}
There exists a constant $c > 0$ such that,
for every integer $h > 1$, the diameter problem in (unweighted, undirected) $n$-vertex
$K_h$-minor-free graphs can be solved in time $\Oh_h(n^{2-c})$. 
\end{conjecture}

\paragraph{Graphs of bounded Euler genus.}
Our main result is the verification of \Cref{conj:taunt} for graphs
of bounded Euler genus. Furthermore, our algorithm computes also the eccentricies
of all the vertices of the input graph $G$. Recall here that
    the eccentricity of a vertex $v \in V(G)$ is defined as
    $\ecc(v) \coloneqq \max_{u \in V(G)} \dist_G(u,v)$, where $\dist_G(\cdot,\cdot)$ is the distance metric in $G$.
\begin{theorem}\label{thm:main-genus}
For every integers $k \geq 1$,
there exists an algorithm that, given an (unweighted, undirected) $n$-vertex graph $G$
of Euler genus at most $k$, runs in time $\Oh_k(n^{2-\frac{1}{25}})$
and computes the diameter of $G$ and the eccentricity of every vertex of $G$.
\end{theorem}

We remark that in~\cite[Section~9]{Cabello19}, Cabello briefly speculated that his approach could be also generalized to graphs embeddable on surfaces of bounded genus. However, as noted in \cite{Cabello19}, this would require significant effort, as the technique works closely on the embedding and in surfaces of higher genus, additional topological hurdles arise. In contrast, in our proof of \Cref{thm:main-genus} the main ingredient is an improved combinatorial bound on the number of so-called \emph{distance profiles}~\cite{LeW24} in graphs of bounded Euler genus. This proof uses topology only very lightly, while the rest of the argument is rather standard and topology-free. All in all, we obtain a robust methodology of approaching the problem, which, as we will see, can be also used to attack \Cref{conj:taunt} to some extent.

To explain our bound on distance profiles, we need to recall several relevant definitions.

Let $G$ be a graph, $R \subseteq V(G)$ be a subset of vertices, and $s_R \in R$ be a vertex in $R$.
The \emph{distance profile} of a vertex $u \in V(G)$ to $R$ (relative to $s_R$)
  is the function $\distprofile{R}{s_R}{u} \colon R \to \mathbb{Z}$ defined as follows:
  \[ \distprofile{R}{s_R}{u}(s) = \dist_G(u,s) - \dist_G(u,s_R)\qquad \textrm{for all }s\in R. \]
Note that provided $R$ is connected\footnote{A subset of vertices $R$ of a graph $G$ is {\em{connected}} if the induced subgraph $G[R]$ is connected.}, we have
$\distprofile{R}{s_R}{u}(s) \in \{-|R|,-|R|+1,\ldots,|R|-1,|R|\}$. In~\cite{LeW24},
Le and Wulff-Nilsen proved that if $R$ is connected and $G$ is $K_h$-minor-free, then
the set system $$\left\{ \left\{(s,i) \in R \times \{-|R|,\ldots,|R|\}~|~i \leq \distprofile{R}{s_R}{u}(s)\right\}~\colon~u \in V(G)\right\}$$ has VC dimension at most $h-1$. Hence, by applying the Sauer-Shelah Lemma we obtain that
\begin{theorem}[\cite{LeW24}]\label{thm:distprofiles-LeW24}
For every integer $h\geq 1$, $K_h$-minor-free graph $G$, connected set $R \subseteq V(G)$, and  $s_R \in R$, there are at most $\Oh_h(|R|^{2h-2})$
different distance profiles to $R$ relative to $s_R$. 
\end{theorem}
The VC dimension argument applied above inevitably leads to a bound with the exponent depending on $h$.
We show that for  graphs of bounded Euler genus, the bound of \Cref{thm:distprofiles-LeW24}
can be improved to a polynomial of degree independent of the parameter.
\begin{theorem}\label{thm:distprofiles}
For every integer $k\geq 1$, (unweighted, undirected) graph $G$ of Euler genus at most~$k$, connected set $R \subseteq V(G)$, and $s_R \in R$,
the number of distance profiles to $R$ relative to $s_R$
is at most~$\Oh_k(|R|^{12})$.
\end{theorem}
The main idea behind the proof of \Cref{thm:distprofiles} is the following simple
observation:
if $P$ is a shortest path from some $u \in V(G)$ to $s_R$, then, as one walks along $P$
from $u$ to $s_R$, the distance profile of the current vertex to $R$ can only (point-wise) increase.
A slightly more technical modification of this argument works for shortest paths
from $u \in V(G)$ to $R$. This allows us to reduce the case of bounded Euler genus graphs
to the planar case by cutting along a constant number of shortest-to-$R$ paths,
and analysing how the distance profiles change during such a process.

One could ask whether an improvement similar to that of \Cref{thm:distprofiles}
would be possible even in the generality of $K_h$-minor-free graphs.
Unfortunately, it seems that \Cref{thm:distprofiles} is the limit of such improvements. More precisely, the following simple example shows that the linear dependency on $h$ in the exponent of the bound on the number of profiles is inevitable even in graphs of treewidth $h$ (which are  $K_{(h+1)^2}$-minor-free).

Let $0 < k \ll \ell$ be positive integers.
Let $R$ be a path of length $\ell$ and $v_1,\ldots,v_k$ be $k$ equidistant
points on $R$ (i.e., the distance between $v_i$ and $v_{i+1}$ is at least $p \coloneqq \lfloor \ell/(k-1) \rfloor$).
For every vector $\mathbf{a} = (a_1,\ldots,a_k) \in \{\ell, \ldots, \ell + p \}^k$,
construct a vertex $u(\mathbf{a})$ and, for every $i\in \{1,\ldots,k\}$, connect it with
$v_i$ using a path of length $a_i$.
This finishes the construction of the graph $G$; see~\Cref{fig:example} for an illustration.
Note that $G$ has treewidth at most $k+1$, because $G-\{v_1,\ldots,v_k\}$ is a forest.
Furthermore, since the distance between consecutive vertices $v_i$ is at least $p$, we have that
  $\dist_G(u(\mathbf{a}), v_i) = a_i$ for every vector $\mathbf{a}$ and $i\in \{1,\ldots,k\}$.
Consequently, if we restrict to vectors $\mathbf{a}$ with $a_1 = \ell$, 
every vertex $u(\mathbf{a})$ has a different distance profile to $R$ relative to $v_1$.
Finally, note that there are $(p + 1)^{k-1} \geq (\ell/(k-1))^{k-1} = \Omega_k(\ell^k)$ different
vectors $\mathbf{a}$ with $a_1=\ell$, giving that many different~profiles.
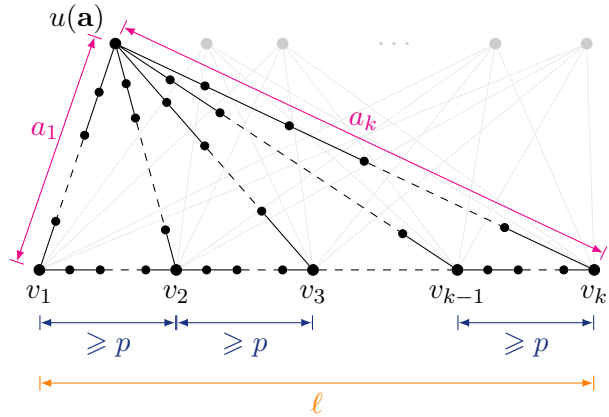
\begin{figure}[ht]
  \centering
  \begin{tikzpicture}
    \tikzset{black node/.style={draw, circle, fill = black, minimum size = 4pt, inner sep = 0pt}}
    \tikzset{small black node/.style={draw, circle, fill = black, minimum size = 3pt, inner sep = 0pt}}

    \node[black node, label={below:$v_1$}] (A) at (0,0) {};
    \node[black node, label={below:$v_2$}] (B) at (1.8,0) {};
    \node[black node, label={below:$v_3$}] (C) at (3.6,0) {};
    \node[black node, label={below:$v_{k-1}$}] (D) at (5.5,0) {};
    \node[black node, label={below:$v_k$}] (E) at (7.3,0) {};

    \node[small black node] at (0.4,0) {};
    \node[small black node] (Al) at (0.8,0) {};
    \node[small black node] (Ar) at (1.4,0) {};

    \node[small black node] at (2.2,0) {};
    \node[small black node] (Bl) at (2.6,0) {};
    \node[small black node] (Br) at (3.2,0) {};

    \node[small black node] at (5.9,0) {};
    \node[small black node] (Dl) at (6.3,0) {};
    \node[small black node] (Dr) at (6.9,0) {};

    \draw (A) -- (Al) (Ar) -- (B)
          (B) -- (Bl) (Br) -- (C)      
          (D) -- (Dl) (Dr) -- (E);
    \draw[dashed] (Al) -- (Ar) (Bl) --  (Br) (C) -- (D) (Dl) -- (Dr);

    \draw[blue] (0,-7mm) coordinate (S) edge[|<->|, >= latex] node[below]{\textcolor{blue}{$\ge p$}} (1.8,-7mm);
    \draw[blue] (1.8,-7mm) coordinate (S) edge[|<->|, >= latex] node[below]{\textcolor{blue}{$\ge p$}} (3.6,-7mm);
    \draw[blue] (5.5,-7mm) coordinate (S) edge[|<->|, >= latex] node[below]{\textcolor{blue}{$\ge p$}} (7.3,-7mm);
    \draw[orange] (0,-15mm) coordinate (S) edge[|<->|, >= latex] node[below]{\textcolor{orange}{$\ell$}} (7.3,-15mm);

    \node[black node, label={[label distance=-2pt]120:$u(\mathbf{a})$}] (X) at (1,3) {};
    
    \node[black node,gray!40!white] (T) at (2.2,3) {};
    \node[black node,gray!40!white] (P) at (3.2,3) {};
    \node[black node,gray!40!white] (Q) at (6,3) {};
    \node[black node,gray!40!white] (R) at (7.2,3) {};
    \node[label={center:\textcolor{gray!40!white}{$\dots$}}] at (4.7,3) {};
    \draw[gray!15!white] (T) -- (A) (T) -- (B) (T) -- (C) (T) -- (D) (T) -- (E);
    \draw[gray!15!white] (P) -- (A) (P) -- (B) (P) -- (C) (P) -- (D) (P) -- (E);
    \draw[gray!15!white] (Q) -- (A) (Q) -- (B) (Q) -- (C) (Q) -- (D) (Q) -- (E);
    \draw[gray!15!white] (R) -- (A) (R) -- (B) (R) -- (C) (R) -- (D) (R) -- (E);

    \path (X) to node[small black node] [pos=0.2] {}  node[small black node] [pos=0.4] (X1) {} 
    node[small black node] [pos=0.8] (X2) {}(A);
    \path (X) to node[small black node] [pos=0.16] {}  node[small black node] [pos=0.32] (Y1) {} 
    node[small black node] [pos=0.84] (Y2) {}(B);
    \path (X) to node[small black node] [pos=0.25] {}  node[small black node] [pos=0.45] (Z1) {} 
    node[small black node] [pos=0.75] (Z2) {}(C);
    \path (X) to node[small black node] [pos=0.15] {}  node[small black node] [pos=0.3] (W1) {} 
    node[small black node] [pos=0.85] (W2) {}(D);
    \path (X) to node[small black node] [pos=0.18] {} node[small black node] [pos=0.36] {}  node[small black node] [pos=0.52] (V1) {} 
    node[small black node] [pos=0.82] (V2) {}(E);

    \draw (X) -- (X1) (X2) -- (A)
    (X) -- (Y1) (Y2) -- (B)
    (X) -- (Z1) (Z2) -- (C)
    (X) -- (W1) (W2) -- (D)
    (X) -- (V1) (V2) -- (E);
    \draw[dashed] (X1) -- (X2)  (Y1) -- (Y2) (Z1) -- (Z2) (W1) -- (W2) (V1) -- (V2);

    \draw[magenta] (0.73,3.1) coordinate edge[|<->|, >= latex] node[above]{\textcolor{magenta}{$a_1$~~~}} (-0.3,0.1);
    \draw[magenta] (1.1,3.22) coordinate edge[|<->|, >= latex] node[above]{\textcolor{magenta}{$a_k$}} (7.43,0.22);

  \end{tikzpicture}
  \caption{Illustration of a construction that shows that linear dependency on $h$ in the exponent of the bound on the number of profiles is inevitable, even in graphs of treewidth $h$.}
  \label{fig:example}
\end{figure}

Our algorithm for \Cref{thm:main-genus} follows closely the approach of Le and Wulff-Nilsen~\cite{LeW24} augmented by the bound provided by \Cref{thm:distprofiles}. Namely, we first compute an $r$-division of the input graph $G$
into regions of size $r=n^{\delta}$, for some small $\delta > 0$. Then we use \Cref{thm:distprofiles}
for individual regions $R$ to speed up the computation of distances between $R$ and $V(G) \setminus R$,
by grouping vertices outside $R$ according to their distance profiles 
to $R$. Each group is batch-processed in a single step.

\paragraph{Generalizations.}
Further, we show that our techniques combine well with the techniques for bounded
treewidth graphs of Cabello and Knauer~\cite{CabelloK09}.
First, we show that \Cref{conj:taunt} holds for classes of graphs
of bounded Euler genus with a constant number of \emph{apices}, i.e., vertices that are arbitrarily connected to the rest of the graph.

\begin{theorem}\label{thm:main-apices}
For every integers $g,k \geq 1$, 
there exists an algorithm that, given an (unweighted, undirected) $n$-vertex graph $G$
and a set $A \subseteq V(G)$ such that $|A| \leq k$ and $G-A$ is
of Euler genus at most $g$, runs in time $\Oh_{g,k}(n^{2-\frac{1}{25}} \log^{k-1} n)$
and computes the diameter of $G$ and the eccentricity of every vertex of $G$.
\end{theorem}
Second, we show that \Cref{conj:taunt} holds for classes of graphs
constructed by clique-sums of graphs as in \Cref{thm:main-apices}.
To state this result formally, we need some definitions. For a graph $G$,
   a \emph{tree decompostion} of $G$ is a pair $(T,\beta)$ where $T$ is a tree
  and $\beta$ is a function that assigns to every $t \in V(T)$ a \emph{bag} 
  $\beta(t) \subseteq V(G)$ such that (1) for every $v \in V(G)$, the set $\{t \in V(T)~|~v \in \beta(t)\}$ is nonempty and connected in $T$, and (2) for every $uv \in E(G)$ there exists $t \in V(T)$ with $u,v \in \beta(t)$. 
 The \emph{torso} of the bag $\beta(t)$ is constructed from $G[\beta(t)]$ by adding, for every
 neighbor $s$ of $t$ in $T$, all edges between the vertices of $\beta(s) \cap \beta(t)$.
\begin{theorem}\label{thm:main-decomp}
For every integer $k \geq 1$, 
there exists an algorithm with the following specification.
The input consists of an (unweighted, undirected) $n$-vertex graph $G$
together with a tree decomposition $(T,\beta)$ of $G$ and a set $A(t) \subseteq \beta(t)$ for
every $t \in V(T)$ satisfying the following properties:
\begin{itemize}[nosep]
\item For every node $t \in V(T)$, we have that $|A(t)| \leq k$ and the torso of $\beta(t)$ with the vertices
of $A(t)$ deleted is a graph of Euler genus at most $k$.
\item For every edge $st \in E(T)$, we have $|\beta(s) \cap \beta(t)| \leq k$.
\end{itemize}
The algorithm runs in time $\Oh_k(n^{2-\frac{1}{356}} \log^{6k} n)$ and computes
the diameter of $G$ and the eccentricity of every vertex of $G$.
\end{theorem}
Note that the statements of
\Cref{thm:main-apices,thm:main-decomp} require the set $A$ and
the decomposition $(T,\beta)$, respectively, to be provided explicitly on input;
this should be compared with more general statements where the algorithm is 
given only $G$ with a promise that such set $A$ or decomposition $(T,\beta)$ exist. At this point, we are not aware of any existing algorithm that would find in subquadratic time a set $A$ as in \Cref{thm:main-apices},
or the decomposition $(T,\beta)$ with the sets $A$ as in \Cref{thm:main-decomp}, even in the approximate sense. However, we were informed by Korhonen, Pilipczuk, Stamoulis, and Thilikos~\cite{KorhonenPST24priv} that it seems likely that the techniques introduced in the recent almost linear-time algorithm for minor-testing~\cite{KorhonenPS24} could be used to construct such an algorithm, with almost linear time complexity. With this result in place, the assumption about the decomposition and/or apex sets being provided on input could be lifted in \Cref{thm:main-apices,thm:main-decomp}; this is, however, left to future work.

\paragraph{Discussion.}
As one of the main outcomes of their theory of graph minors, Robertson and Seymour proved the following Structure Theorem~\cite{RobertsonS03a}: every $K_h$-minor-free graph $G$
admits a tree decomposition $(T,\beta)$ such that
\begin{itemize}[nosep]
 \item for every pair $s,t$ of adjacent nodes of $T$, the set $\beta(t) \cap \beta(s)$ has size $\Oh_h(1)$; and
 \item the torso of every bag $\beta(t)$ is ``nearly embedable'' into a surface of bounded (in terms of $h$) Euler~genus.
\end{itemize}
The notion of being ``nearly embeddable'' encompasses adding a constant number of apices (which can be handled by \Cref{thm:main-decomp}) and a constant number of so-called
vortices (which are not handled by \Cref{thm:main-decomp}). Thus, our methods fall short of verifying  \Cref{conj:taunt} in full generality due to vortices.

We remark that recently, Thilikos and Wiederrecht~\cite{ThilikosW22} proved a variant of the Structure Theorem, where under the stronger assumption of excluding a minor of a {\em{shallow vortex grid}}, instead of a clique minor, they gave a decomposition as above, but with torsos devoid of vortices. Thus, the decomposition for shallow-vortex-grid-minor-free graphs provided by~\cite{ThilikosW22} can be directly plugged into \Cref{thm:main-decomp}, with the caveat that~\cite{ThilikosW22} does not provide a subquadratic algorithm to compute the decomposition.

Coming back to \Cref{conj:taunt},
the simplest case
that we are currently unable to solve is the setting when the input is a planar graph plus a single vortex. More formally, for a fixed integer $k$, let $\mathcal{G}_k$ be the class of graphs
defined as follows. We have $G \in \mathcal{G}_k$ if there exist two subgraphs $G_0,G_1$
of $G$ and a sequence of vertices $v_1,\ldots,v_b$ in $V(G_0) \cap V(G_1)$
such that:
\begin{itemize}[nosep]
 \item $V(G) = V(G_0) \cup V(G_1)$,
 \item $E(G) = E(G_0) \cup E(G_1)$,
 \item $G_0$ admits a planar embedding where the vertices $v_1,\ldots,v_b$ lie on one face in
this order, and
\item $G_1$ admits a tree decomposition $(T_1,\beta_1)$, where $T_1$ is a path on nodes $t_1,\ldots,t_b$ and for every $i\in\{1,\ldots,b\}$, the bag $\beta_1(t_i)$ contains $v_i$
and is of size at most $k$.
\end{itemize}
It is easy to see that graphs from $\mathcal{G}_k$ are $K_{k+\Oh(1)}$-minor-free.
Do they satisfy \Cref{conj:taunt}? That is, is there a constant
$c > 0$ such that the diameter problem in $\mathcal{G}_k$ can be solved in
time $\Oh_k(n^{2-c})$?

\paragraph{Organization.}
We prove \Cref{thm:distprofiles} in \Cref{sec:distprofiles}.
\Cref{thm:main-apices} is proven in \Cref{sec:algo-genus}; note that \Cref{thm:main-genus} follows from \Cref{thm:main-apices}
for $k=1$. 
\Cref{thm:main-decomp} is proven in \Cref{sec:algo}.

\section{Preliminaries}\label{sec:prelims}

\paragraph{Set systems and VC-dimension.}
A \emph{set system} is a collection $\mathcal{F}$ of subsets of a given set $A$, which we call \emph{ground set} of $\mathcal{F}$.
We say that a subset $Y\subseteq A$ is \emph{shattered by} $\mathcal{F}$
if $\{Y \cap S : S \in \mathcal{F}\} = 2^Y$, that is, the intersections of $Y$ and the sets in $\mathcal{F}$ contain every subset of $Y$.
The \emph{VC-dimension} of a set system $\mathcal{F}$ with ground set $A$
is the size of the largest subset $Y\subseteq A$ shattered by $\mathcal{F}$. The notion of VC-dimension was introduced by Vapnik and Chervonenkis~\cite{VapnikC71}.

We will use the following well-known Sauer-Shelah Lemma~\cite{Sauer72,Shelah72}, which gives a polynomial upper bound on the size of a set system of bounded VC-dimension.

\begin{lemma}[Sauer-Shelah Lemma]\label{lem:sauer-shelah}
  Let $\mathcal{F}$ be a set system with ground set $A$.
  If the VC-dimension of $\mathcal{F}$ is at most $k$, then $|\mathcal{F}|=\mathcal{O}(|A|^k)$.
\end{lemma}

\paragraph{Basic graph notation.} All our graphs are undirected.
For a graph $G$, the neighborhood of a vertex $u$ is defined as $N_G(u) = \{v~\colon~uv \in E(G)\}$
and for $X \subseteq V(G)$ we have $N_G(X) = \bigcup_{u \in X} N_G(u) \setminus X$. 

The {\em{length}} of a path $P$, denoted $|P|$, is the number of edges of $P$.
For two vertices $u,v$ of a graph $G$, the {\em{distance}} between $u$ and $v$,
denoted $\dist_G(u,v)$, is defined as the minimum length of a path in $G$ with endpoints $u$ and $v$.
For every $v\in V(G)$ and set $R \subseteq V(G)$, we set $\dist_G(v,R)\coloneqq \min\{\dist_G(v,y)\colon y\in R\}.$
For vertices $x,y$ appearing on a path $P$, by $P[x,y]$ we denote the subpath of $P$ with endpoints $x$ and $y$.
The set of vertices traversed by a path $P$ is denoted by $V(P)$.
In all above notation, we sometimes drop the subscript if the graph is clear from the context.

For a nonnegative integer $q$, we use the shorthand $[q]\coloneqq\{1,\ldots,q\}$.
For a vertex $v \in V(G)$ and a set $X \subseteq V(G)$, we define the \emph{$X$-eccentricity} of $v$ as
$\ecc_X(v) \coloneqq \max_{x \in X} \dist(v,x)$.
Thus, the eccentricity of $v$ in $G$ is the same as its $V(G)$-eccentricity.

The \emph{Euler genus} of a graph $G$ is the minimum Euler characteristic of a surface, where $G$ is embeddable.
We refer to the textbook of Mohar and Thomassen for more on surfaces and embedded graphs~\cite{MoharT01grap}.

We will use the following result of Le and Wulff-Nilsen~\cite[Theorem 1.3]{LeW24} for planar graphs. Note that the set $R$ is not necessarily connected.

\begin{theorem}\label{thm:LW}
  Let $h\ge 1$ be an integer, $G$ be a $K_h$-minor-free (unweighted, undirected) graph, $R$ be a subset of $V(G)$, and $s_R\in R$. Then the set system $$\left\{\left\{(s,i) \in R \times \mathbb{Z}~|~i \leq \dist_G(u,s)-\dist_G(u,s_R)\right\}~\colon~u \in V(G)\right\}$$
  has VC-dimension at most $h-1$.
\end{theorem}

\paragraph{Algorithmic tools.}
All our algorithms assume the word RAM model.

To cope with apices, we will need the following classic data structure due to Willard~\cite{Willard85}.
\begin{theorem}[\cite{Willard85}]\label{t:orth_query}
Let $V$ be a set of $n$ points in $\mathbb{R}^d$ and let $w\colon V \to \mathbb{R}$ be a weight function. By a \emph{suffix range}, we mean any set of the form
$$\mathsf{Range}(r_1,\ldots,r_d)\coloneqq \{ (x_1, \dots, x_d) \in \mathbb{R}^d\,\mid\,x_i \geq r_i\textrm{ for all }i\in [d] \}$$ for some range parameters $r_1,\ldots,r_d\in \mathbb{R}$.

There is a data structure that uses $\Oh\left(n \log^{d - 1} n \right)$ preprocessing time, $\Oh\left( n \log^{d - 1} n \right)$ memory and answers the following suffix range queries in time $\Oh \left( \log^{d - 1} n \right)$:
given a tuple $(r_i)_{i \in [d]}$, find the maximum value of $w(v)$ over all
$v \in V\cap \mathsf{Range}(r_1,\ldots,r_d)$.
\end{theorem}

We will also need the following standard statement about $r$-divisions.

\begin{theorem}[\cite{WulffNilsen11}]\label{t:r_division}
Let $G$ be a $K_t$-minor-free graph on $n$ vertices. For any fixed constant $\varepsilon > 0$, and for any parameter $r$ with $C t^2 \log n\leq r\leq  n$, where $C$ is some absolute constant, we can construct in time $\Oh \left( n^{1 + \varepsilon} \sqrt{r} \right)$ an $r$-division of $G$, that is, a collection $\mathcal{R}$ of connected subsets of vertices of $G$ such that:
\begin{itemize}[nosep]
\item $\bigcup \mathcal{R}=V(G)$,
\item $|R| \leq r$ for every $R \in \mathcal{R}$, and
\item $\sum_{R \in \mathcal{R}} |\partial R| \leq \Oh(nt / \sqrt{r})$,
where $\partial R = R \cap N_G(V(G) \setminus R)$.
\end{itemize}
\end{theorem}

\section{Distance profiles in graphs of bounded Euler genus}\label{sec:distprofiles}

In this section we prove~\Cref{thm:distprofiles}. Our argument consists of a reduction to the planar case, where we can use the constant bound on the VC-dimension of the set system given by the distance profiles due to Le and Wulff-Nilsen~\cite{LeW24}.
The main idea behind the reduction is to consider certain notions of ``extended'' profiles, where the extension is built along a collections of shortest paths. These shortest paths can be chosen in such a way that by cutting the graph along these paths we obtain a plane graph. Then a bound on the number of the extended profiles in the obtained plane graph translates to a bound on the number of (standard) distance profiles in the original graph.

Preliminary definitions and results needed for defining profiles with respect to shortest paths are given in~\Cref{subsec:milestones}.
These extended profiles are then defined in~\Cref{subsec:mdprofiles}. There, we also prove that a fundamental lemma that equality of extended profiles entails equality of (standard) distance profiles.
The main reduction providing the proof of~\Cref{thm:distprofiles} is given at the end of this section.

\subsection{Milestones}
\label{subsec:milestones}

Let $G$ be a graph, $R$ be a subset of $V(G)$,
$v_0$ be a vertex in $V(G)$, and $P$ be a shortest path from $v_0$ to $R$. Let $x$ be the unique vertex in $V(P)\cap R$. Further, let $\le_P$ be the linear ordering of the vertices traversed by $P$: for two vertices $v,u\in V(P)$, we have $v\le_P u$ if $u$ belongs to $P[v,x]$.
We say that a vertex $v\in V(P)$ is a \emph{milestone of $P$}
if either $v=x$ or we have $\distprofile{R}{x}{v}\neq\distprofile{R}{x}{u}$, where $u$ is the successor of $v$ in $\le_P$.
We denote by $M_{R}(P)$ the set of all milestones of $P$.
Given a milestone $v\in M_{R}(P)$,
the \emph{neutral prefix of $v$ in $P$} is defined as the vertex set of the maximal subpath $Q$ of $P[v_0,v]$ satisfying the following: $v$ is the only milestone of $P$ that belongs to $Q$.

The next lemma shows that minimum-length paths towards $R$ that contain a vertex in the neutral prefix of a milestone can be assumed to pass through that milestone vertex.

\begin{lemma}
  \label{lem:dispref}
  Let $G$ be a graph, $R$ be a subset of $V(G)$, $v_0$ be a vertex in $V(G)$ and $P$ be a shortest path from $v_0$ to $R$.
  Then for every $v\in M_R(P)$, every $u$ in the neutral prefix of $v$, and every $y\in R$, it holds that
  $\dist(u,y)=|P[u,v]|+\dist(v,y)$. 
\end{lemma}
\begin{proof}
  Let $x$ be the unique vertex of $V(P)\cap R$. Note that, by definition, $\distprofile{R}{x}{v}=\distprofile{R}{x}{u}$.
  Also, $\dist(u,x)=\dist(u,v)+\dist(v,x)$ and $\dist(u,v)=|P[u,v]|$. Therefore, $\dist(u,y)=|P[u,v]|+\dist(v,y)$ for every $y\in R$.
\end{proof}

We also give an upper bound on the number of milestones.

\begin{lemma}
  \label{lem:mile}
  Let $G$ be a graph, $R$ be a connected subset of $V(G)$, $v_0$ be a vertex of $G$, and $P$ be a shortest path from $v_0$ to $R$. Then the number of milestones of $P$ is at most $|R|^2+1$.
\end{lemma}
\begin{proof}
  Let $x$ be the unique vertex of $V(P)\cap R$.
  First observe that since $P$ is a shortest path from $v_0$ to~$R$, we have $\dist(v,y)\ge \dist(v,x)$ for every $v\in V(P)$ and every $y\in R$; hence $\distprofile{R}{x}{v}(y)\ge 0$.
  Also, since $R$ is connected, for every $y\in R$ we have $\distprofile{R}{x}{x}(y)\le |R|$.
  To conclude the proof, it suffices to prove that for all $v_1,v_2\in V(P)$ with $v_1\leq_P v_2$, we have \begin{equation}\label{eq:wydra}\distprofile{R}{x}{v_1}(y)\le \distprofile{R}{x}{v_2}(y)\qquad\textrm{for all }y\in R.\end{equation}
  Indeed, \eqref{eq:wydra} together with the previous observations shows that all the distinct distance profiles of the form $\distprofile{R}{x}{v}$ for $v\in V(P)$ can be treated as vectors of length $|R|$ with entries in $\{0,\ldots,|R|\}$, and they all have distinct sums $\sum_{y\in R} \distprofile{R}{x}{v}(y)$. Since these sums range between $0$ and $|R|^2$, the total number of distinct profiles is at most $|R|^2+1$, implying the same bound on the number of milestones.

  To see why \eqref{eq:wydra} holds, note that $\dist(v_1,y)\le \dist(v_1,v_2)+\dist(v_2,y)$ implies that $$\dist(v_1,y)\le \dist(v_1,v_2)+\distprofile{R}{x}{v_2}(y)+\dist(v_2,x)=\distprofile{R}{x}{v_2}(y) + \dist(v_1,x);$$
  the last equality follows from $P$ being a shortest path containing $v_1,v_2,$ and $x$ (in this order). This in turn implies that
  $\distprofile{R}{x}{v_1}(y)=\dist(v_1,y)-\dist(v_1,x)\le \distprofile{R}{x}{v_2}(y)$, as claimed.
\end{proof}

\subsection{Anchor-distance profiles}
\label{subsec:mdprofiles}

\paragraph{Shortest path collections.}
Let $G$ be a graph and $R$ be a subset of vertices of $G$.
We say that a collection $\mathcal{P}$ of paths in $G$
is an \emph{$R$-shortest path collection} if
\begin{itemize}[nosep]
  \item every $P \in \mathcal{P}$ is a shortest path from some $v^P \in V(G)$ to $R$, i.e., $|P|=\dist(v^P,R)$; and
  \item $R \subseteq \bigcup_{P \in \mathcal{P}} V(P)$.
\end{itemize}

For each $P \in \mathcal{P}$, we denote by $x^P$ the endpoint of $P$ in $R$.
Note that the collection $\mathcal{P}$ obtained by taking, for every $y\in R$, the zero-length path from $y$ to $y$, is an $R$-shortest path collection. 

We say that an $R$-shortest path collection is \emph{consistent} if, for every $P_1,P_2 \in \mathcal{P}$
and $v \in V(P_1) \cap V(P_2)$ the paths $P_1[v,x^{P_1}]$ and $P_2[v,x^{P_2}]$ are equal. That is, 
once two paths intersect, they continue together towards $R$. 

The following statement is a direct consequence of the definition of an $R$-shortest path collection.
\begin{observation}\label{obs:short}
  Let $G$ be a graph, $R$ be a subset of vertices of $G$, and $\mathcal{P}$ be an $R$-shortest path collection. Then for every two paths $P_1,P_2\in \mathcal{P}$ and every $v\in V(P_1)\cap V(P_2)$, we have $|P_1[v,x^{P_1}]|=|P_2[v,x^{P_2}]|$.
\end{observation}

\paragraph{Anchor vertices and their prefixes.}
Let $G$ be a graph, $R$ be a subset of $V(G)$, and $\mathcal{P}$ be an $R$-shortest path collection. We denote by $H_{\mathcal{P}}$ the union of the paths in $\mathcal{P}$, i.e., the graph $(\bigcup_{P\in\mathcal{P}}V(P),\bigcup_{P\in\mathcal{P}}E(P))$.
We say that a vertex is an \emph{anchor vertex} if either it has degree more than two in $H_{\mathcal{P}}$
or it is a milestone of a path $P \in \mathcal{P}$.
We denote by $A_R(P)$ the set of all anchor vertices lying on a path $P \in \mathcal{P}$
and by $A_R(\mathcal{P})$ the set of all anchor vertices for $\mathcal{P}$.
Given a path $P\in\mathcal{P}$ with endpoints $v_0$ and $y\in R$, and an anchor vertex $w\in A_R(P)$,
the \emph{prefix of $w$ in $P$} is the vertex set of the maximal subpath $Q$ of
$P[v_0,v]$ satisfying the following: $v$ is the only anchor vertex of $P$ that belongs to $Q$.
Note that for every anchor $w\in V(P)$ there is a milestone $w'$ of $P$ (possibly $w=w'$)
such that the prefix of $w$ in $P$ is a subset of the neutral prefix of $w'$ in $P$.
Finally, for an anchor vertex $w$, the \emph{tail} of $w$, 
denoted $\mathrm{tail}(w)$, is the subgraph of $G$ consisting
of the union of all prefixes of $w$ in $P$ over all paths $P \in \mathcal{P}$ that contain $w$.

\paragraph{Hat-distances.}
Let $G$ be a graph, $R$ be a subset of vertices of $G$, and $\mathcal{P}$ be an $R$-shortest path collection.
We denote by $$U_{\mathcal{P}}\coloneqq V(G)-\bigcup_{P\in\mathcal{P}}V(P).$$
For every $u\in U_{\mathcal{P}}$,
and every anchor vertex $w\in A_R(\mathcal{P})$,
we set 
\[ \widehat{\dist}(u,w)\coloneqq \min\{|Q_{u,z}|+|P[z,w]|\colon
P \in \mathcal{P} \wedge w \in V(P) \wedge 
z\text{ is in the prefix of $w$ in $P$}\},\] 
where $Q_{u,z}$ is a shortest path from $u$ to $z$ with all its internal vertices in $U_{\mathcal{P}}$.
If such $Q_{u,z}$ does not exist for any $z \in V(\mathrm{tail}(w))$,
we set $\widehat{\dist}(u, w) \coloneqq \infty$. 

The following statement is a direct consequence of the definition of $\widehat{\dist}(\cdot,\cdot)$.

\begin{observation}
  \label{obs:dist}
  Let $G$ be a graph, $R$ be a subset of vertices of $G$, and $\mathcal{P}$ be
  an $R$-shortest path collection.
  Then for every $u\in U_{\mathcal{P}}$, we have that $$\dist(u,R)=\min \left\{\widehat{\dist}(u,w)+ \dist(w,R)
  \colon w\in A_R(\mathcal{P})\right\}.$$
\end{observation}

\paragraph{Anchor-distance profiles.}
Let $G$ be a graph, $R$ be a subset of vertices of $G$,  and $\mathcal{P}$
be an $R$-shortest path collection.
The \emph{anchor-distance profile} of a vertex $u\in U_{\mathcal{P}}$ to $R$ with respect to $\mathcal{P}$
is a function $\diststarprofile{R}{\mathcal{P}}{u}$ mapping each $w \in A_R(\mathcal{P})$ to
\[\diststarprofile{R}{\mathcal{P}}{u}(w) \coloneqq \widehat{\dist}(u,w) + \dist(w,R)
 - \dist(u,R).\]
\Cref{obs:dist} implies that we always have $\diststarprofile{R}{\mathcal{P}}{u}(w) \geq 0$.
We set
\[\hatprofile{R}{\mathcal{P}}{u}(w)\coloneqq\min\{\diststarprofile{R}{\mathcal{P}}{u}(w),|R|+1\}.\]

\begin{lemma}\label{lem:hat-to-normal}
  Let $G$ be a graph, let $R$ be a connected subset of vertices of $G$, and $s_R\in R$. Also, let $\mathcal{P}$ be an $R$-shortest path collection.
  Then for all $u_1,u_2\in U_{\mathcal{P}}$,
  $$\hatprofile{R}{\mathcal{P}}{u_1}=\hatprofile{R}{\mathcal{P}}{u_2}\qquad\textrm{implies}\qquad \distprofile{R}{s_R}{u_1}=\distprofile{R}{s_R}{u_2}.$$
\end{lemma}
\begin{proof}
  Fix $u_1,u_2\in U_{\mathcal{P}}$ with
  $\hatprofile{R}{\mathcal{P}}{u_1}=\hatprofile{R}{\mathcal{P}}{u_2}$.
  We start by proving the following.

  \begin{claim}\label{cl:bnd}
    Let $u\in U_{\mathcal{P}}$ and $y\in R$. There is an anchor $w \in A_R(\mathcal{P})$
    such that
    \begin{itemize}[nosep]
      \item $\widehat{\dist}(u,w)+\dist(w,y) = \dist(u,y)$ and
      \item $\hatprofile{R}{\mathcal{P}}{u}(w)\le |R|$.
    \end{itemize}
  \end{claim}
  \begin{proof}
    Let $Q$ be a shortest path from $u$ to $y$ and let $P\in\mathcal{P}$ be the path which $Q$ first intersects
    (if the first vertex of $Q$ in $\bigcup_{P \in \mathcal{P}} V(P)$ belongs to more than one paths in~$\mathcal{P}$,
    we choose $P$ arbitrarily among these paths). 
    Also, let $u'$ be the first vertex of $Q$ (when ordering from $u$ to $y$) in $V(P)$
    and $w$ be the anchor of $P$ that contains $u'$ in its prefix (in $P$).
    Note that $u' \in V(\mathrm{tail}(w))$.

    We first show that 
    \begin{equation}\label{eq:bnd:1}\widehat{\dist}(u,w)+\dist(w,y)= \dist(u,y).\end{equation}
    By~\Cref{lem:dispref} and the fact that $|Q[u',y]| = \dist(u',y)$, we have
    \begin{equation}\label{eq:I}
      \dist(w,y) = |Q[u',y]|-|P[u',w]|.
    \end{equation}
    Also, by definition, we have
    \begin{equation}\label{eq:II}
      \widehat{\dist}(u,w)\le |Q[u,u']|+|P[u',w]|.
    \end{equation}
    By~\eqref{eq:I} and~\eqref{eq:II}, we get that $\widehat{\dist}(u,w)+\dist(w,y)\le |Q|$.
    Moreover, since $Q$ is a shortest path from $u$ to $y$ and $\widehat{dist}(u,w) \geq \dist(u,w)$,
    we have
    \[ |Q| = \dist(u,y) \leq \dist(u,w) + \dist(w,y) \leq \widehat{dist}(u,w) + dist(w,y).\]
    This proves~\eqref{eq:bnd:1}.
 
    Next, we show that $\hatprofile{R}{\mathcal{P}}{u}(w)\le |R|$. Note that
    \begin{align*}
      \diststarprofile{R}{\mathcal{P}}{u}(w) + \dist(u,R) & = \widehat{\dist}(u,w)+\dist(w,R)\\
      & \le \widehat{\dist}(u,w)+\dist(w,y)= \dist(u,y).
    \end{align*}
    The connectivity of $R$ implies that $\dist(u,y)\le \dist(u,R)+|R|$, which gives
    $\diststarprofile{R}{\mathcal{P}}{u}(w)\le |R|$, and the claim follows.
  \end{proof}

  We next show that
  there is an integer $c$ such that for every $y\in R$, we have
  $$\dist(u_1,y)=\dist(u_2,y)+c.$$
  Note that this will immediately imply that
  $\distprofile{R}{s_R}{u_1}=\distprofile{R}{s_R}{u_2}$.

  By Observation~\ref{obs:dist}, for every $h \in \{1,2\}$, there is an anchor $w_h \in A_R(\mathcal{P})$
  such that $\dist(u_h,R) = \widehat{\dist}(u_h,w_h) + \dist(w_h, R)$, which is equivalent to
  $\diststarprofile{R}{\mathcal{P}}{u_h}(w_h) = 0$. 
  If $w_h$ lies on $P_h \in \mathcal{P}$, then $\dist(u_h,R) = \dist(u_h,x^{P_h})$. 
  Therefore, as $\diststarprofile{R}{\mathcal{P}}{u_1}=\diststarprofile{R}{\mathcal{P}}{u_2}$,
  we can choose $w_1 = w_2$ and $P_1 = P_2$, hence $x^{P_1} = x^{P_2}$. 
  In other words, there exists $x\in R$ such that $\dist(u_1,R) = \dist(u_1,x)$ and $\dist(u_2,R) = \dist(u_2,x)$.
  We set $c\coloneqq \dist(u_1,x)-\dist(u_2,x)=\dist(u_1,R)-\dist(u_2,R)$.

  Now, fix $y\in R$.
  Let $w_1\in A_R(P)$ be the anchor from~\Cref{cl:bnd} (applied for $u_1$ and $y$).
  As $\hatprofile{R}{\mathcal{P}}{u_1} = \hatprofile{R}{\mathcal{P}}{u_2}$ and
  $\diststarprofile{R}{\mathcal{P}}{u_1}(w_1)\le |R|$, we have
   $\diststarprofile{R}{\mathcal{P}}{u_1}(w_1) = \diststarprofile{R}{\mathcal{P}}{u_2}(w_1)$, i.e.,
  \[\widehat{\dist}(u_1,w_1) + \dist(w_1,R) - \dist(u_1,R) = \widehat{\dist}(u_2,w_1) + \dist(w_1,R)- \dist(u_2,R).\]
  Therefore,
  \begin{align*}
    \dist(u_1,y) &= \widehat{\dist}(u_1,w_1) + \dist(w_1,y)\\
      & =\widehat{\dist}(u_2,w_1) + \dist(w_1,y)+c\geq \dist(u_2,y)+c;
  \end{align*}
  the first equality follows from~\Cref{cl:bnd}.
  Thus $\dist(u_2,y)+c\le \dist(u_1,y)$.
  A symmetric reasoning shows that also $\dist(u_1,y)-c\le\dist(u_2,y)$.
  Therefore we get $\dist(u_1,y) = \dist(u_2,y)+c$, as required.
\end{proof}

\subsection{Reduction from bounded genus graphs to planar graphs}

We next recall several definitions related to embeddings of graphs on surfaces.
Our basic terminology follows~\cite{MoharT01grap}.
We say that a graph $H$ embedded in a surface $\Sigma$ is a {\em{simple cut-graph}} of $\Sigma$ if $H$ has a single face that is also homeomorphic to an open disk; equivalently, $H$ has a single facial walk.
Given a surface $\Sigma$ and a simple cut-graph $H$ on $\Sigma$, we denote by $\Sigma\cutgraph H$ the surface obtained by cutting $\Sigma$ along~$H$. Note that, provided $H$ is a simple cut-graph, $\Sigma\cutgraph H$ is always a disk.

Suppose now that a graph $G$ embedded on  $\Sigma$ and $H$ is a subgraph of $G$ that is a simple cut-graph of $H$.
We define $G\cutgraph H$ to be the graph embedded on $\Sigma\cutgraph H$ obtained from $G$ as follows.
First, let $\sigma$ be the (unique) facial walk of $H$ and note that each edge $e$ of $H$ is contained exactly twice in $\sigma$ and each vertex $v$ of $H$ is contained in $\sigma$ as many times as the degree of $v$ in $H$.
To obtain $G\cutgraph H$, we replace $H$ with a simple cycle $C_\sigma$ whose vertex set is the set of copies of vertices of $H$ and its edge set is the set of copies of edges of $H$ in the obvious way. Notice that $\sigma$ also prescribes for every edge $uv$ of $G$ between a vertex $u\in V(G)\setminus V(H)$ and a vertex $v\in V(H)$, to which copy of $v$ in $G\cutgraph H$ the vertex $u$ should be adjacent to (in $G\cutgraph H$).
See~\Cref{fig:cutopen} for an illustration.

\begin{figure}[ht]
  \centering
  \includegraphics[width=0.8\textwidth]{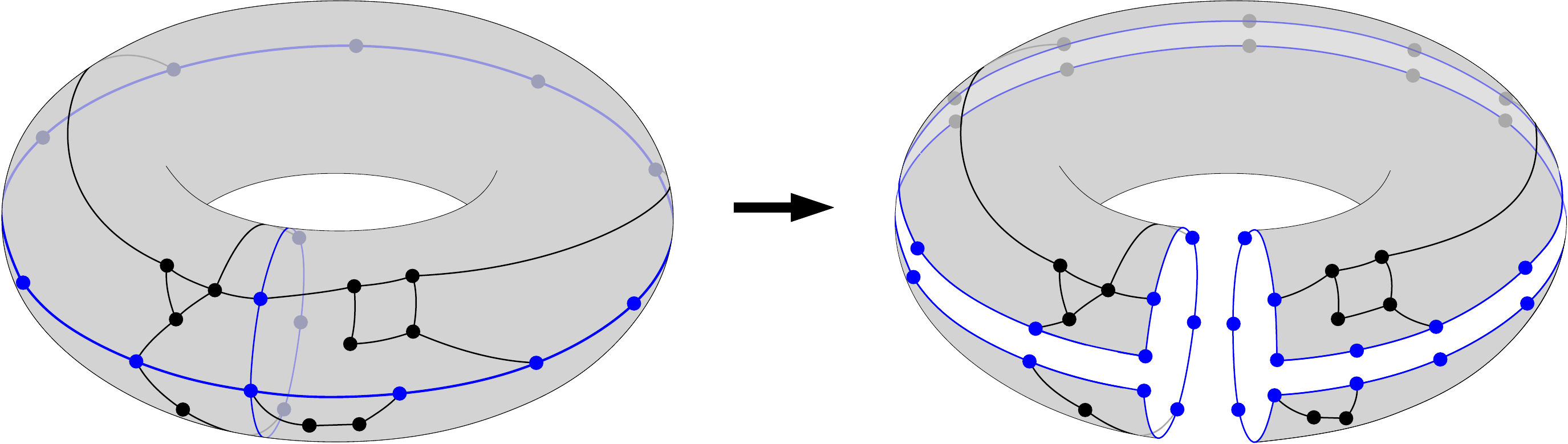}
  \caption{Left: A graph $G$ embedded on a surface $\Sigma$ and a subgraph $H$ of $G$ (in blue) that is a simple cut-graph of $\Sigma$. Right: The graph $G\cutgraph H$ embedded on the surface $\Sigma\cutgraph H$ (which is homeomorphic to a disk); the blue vertices/edges are copies of the vertices/edges of $H$.}
  \label{fig:cutopen}
\end{figure}

We will use the following well-known result which appears in the literature under different formulations; see e.g.~\cite{BorradaileDT14,CabelloCL12algo,EricksonW05}.

\begin{lemma}\label{lem:genuscut}
  For every integer $k\ge 1$ and for every edge-weighted connected graph $G$ embedded on a surface $\Sigma$ of Euler genus at most $k$ and every vertex $u\in V(G)$, there is a subgraph $H$ of $G$ with the following properties:
  \begin{itemize}[nosep]
    \item  $H$ is a simple cut-graph of $\Sigma$, and
    \item  $V(H)$ is the union of the vertex sets of $\mathcal{O}(k)$ shortest paths in $G$ that have $u$ as a common endpoint.
  \end{itemize}
  \end{lemma}

We are now ready to proceed to the proof of~\Cref{thm:distprofiles}.
Employing~\Cref{lem:hat-to-normal},
we aim at bouding the VC-dimension of the set system defined by the anchor-distance profiles.
This is can be done by a suitable reduction to the planar setting using~\Cref{lem:genuscut}.

\begin{proof}[Proof of~\Cref{thm:distprofiles}]
  We assume that $G$ is connected -- the distance profiles of all vertices that are not connected to $R$ are equal.
  Let $T_R$ be a spanning tree of $G[R]$ and let $G_0$ be the graph obtained from $G$ after contracting $T_R$ into a single vertex $v_R$. 
  Consider an embedding of $G_0$ on a surface $\Sigma$ of Euler genus at most $k$.
  By~\Cref{lem:genuscut}, there is a subgraph $H_0$ of $G_0$ that is a simple cut-graph of $\Sigma$ and a family $\mathcal{P}_0$ of $\mathcal{O}(k)$ shortest
  paths in $G_0$, each with $v_R$ as an endpoint, such that $V(H_0)=\bigcup_{P\in\mathcal{P}_0}V(P)$.
  Furthermore, as Lemma~\ref{lem:genuscut} handles edge weights, we can slightly perturb
  the weights so that shortest paths in $G_0$ are unique and, in particular, 
  all shortest paths with one endpoint in $v_R$ form a tree.
  Since $H_0$ is a simple cut-graph of $\Sigma$, $G_0\cutgraph H_0$ is disk-embedded.
  Uncontracting $T_R$, we get a subgraph $H$ of $G$ such that $G\cutgraph H$ is disk-embedded.
  Let $\mathcal{P}$ be the family of projections of the paths of $\mathcal{P}_0$
  onto $G$ plus, for every $y \in R$, a zero-length path from $y$ to $y$. 
  Hence, $\mathcal{P}$ is an $R$-shortest paths collection of size $\mathcal{O}(k)$
  with $V(\mathcal{P}) = V(H)$.
  Furthermore, since in $G_0$ the paths of $\mathcal{P}_0$ formed a tree rooted
  at $v_R$, $\mathcal{P}$ is consistent.

  Note that due to~\Cref{lem:mile} we have that $\sum_{P\in\mathcal{P}} |M_R(P)|\le \mathcal{O}_k(|R|^2)$.
  Also, since $\mathcal{P}$ is consistent, 
  if $B$ are the vertices that are not in $R$
  (recall that vertices in $R$ are milestones) and have degree more than two in the graph obtained by the union of the paths in $\mathcal{P}$, then $|B|\leq |\mathcal{P}|-1$.
  Hence,
  \begin{equation}
    \sum_{P\in\mathcal{P}} |A_R(P)|\le \mathcal{O}_k(|R|^2).\label{eq:mile}
  \end{equation}
  We set $\mathcal{T}$ be the set of all vertices of $G\cutgraph H$
  that are copies of the anchor vertices $A_R(\mathcal{P})$.
  Every anchor vertex has $\mathcal{O}_k(1)$ copies in $\mathcal{Q}$
  and therefore, due to~\eqref{eq:mile},
  \begin{equation}
    |\mathcal{T}|=\mathcal{O}_k(|R|^2).\label{eq:term_size}
  \end{equation}
  For $s \in \mathcal{T}$, let $w(s) \in A_R(\mathcal{R})$ be the anchor vertex
  whose copy (in $G \cutgraph H$) is $s$. In the other direction, for $w \in A_R(\mathcal{R})$, let $S(w)$
  be the set of copies of $w$ in $G \cutgraph H$. 
  
  Let $U$ be the set of vertices of $G\cutgraph H$ that are \textsl{not} copies of vertices from $H$ (i.e., $U=V(G)\setminus V(H)$).
  We set $E_{\mathsf{out}}$ be the set of all edges $uv$ of $G\cutgraph H$ where $u\in U$ and $v$ is a copy of a vertex from $H$,
  i.e., $v\in V(G\cutgraph H)\setminus U$.
  We also set $E_{\mathsf{next}}$ be the set of all edges $uv$ of $G\cutgraph H$ where $u$ is a copy
  of an anchor vertex $w\in A_R(P)$ for some $P\in\mathcal{P}$ and $v$ is a copy of the neighbor of $w$
  in $P$ that \textsl{is not} in the prefix of $w$ in $P$.
  
  Let now $\widehat{G}$ be the graph obtained from $G\cutgraph H$ after the following modifications:
  \begin{itemize}[nosep]
    \item we subdivide $|V(G)|$-many times each edge in $E_{\mathsf{out}}\cup E_{\mathsf{next}}$,
    \item we introduce a new vertex $t$ and add, for every $s\in\mathcal{T}$, a path between $t$
    and $s$ of length $$d_{w(s),t}\coloneqq |V(G)| + \dist_G(w(s),R).$$
  \end{itemize}
  See~\Cref{fig:hatG}.
  Observe that since $G\cutgraph H$ is disk-embedded, $\widehat{G}$ is planar,
  because we may embed $t$ together with all the added paths outside of the disk containing $G\cutgraph H$.

  \begin{figure}[ht]
    \centering
    \includegraphics[width=0.3\textwidth]{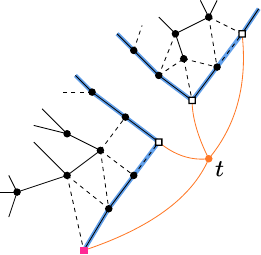}
    \caption{An illustration of (a part of) the construction of the graph $\widehat{G}$. The squared vertices are copies of anchor vertices. The marked squared vertex is also a copy of a vertex in $R$. The highlighted edges are copies of edges of $H$ in $G\cutgraph H$, while the paths obtained by subdividing the edges of $E_{\mathsf{out}}\cup E_{\mathsf{next}}$ are depicted with dashed edges. Edges adjacent to $t$ correspond to paths of appropriate length.}
    \label{fig:hatG}
  \end{figure}

  For every $u\in U$, we define a function $\pi[u]$, mapping every $w\in A_R(\mathcal{P})$ to
  \[\pi[u](w)\coloneqq\min\{\dist_{\widehat{G}}(u,s): s \in S(w)\}+d_{w,t}-\dist_{\widehat{G}}(u,t).\]
  Also, we set $\widehat{\mathcal{X}}\coloneqq\{\widehat{X}_u\mid u\in U\}$,
  where for $u\in U$,
  $$\widehat{X}_u\coloneqq\left\{(w,i) \in A_R(\mathcal{P}) \times \{0,\ldots,|R|+1\}~|~i \leq \pi[u](w)\right\}.$$

  \begin{claim}\label{cl:vcd}
    The set system $\widehat{\mathcal{X}}$ has size $\mathcal{O}_k(|R|^{12})$.
  \end{claim}

  \begin{proof}
  We set $\mathcal{T}^+\coloneqq \mathcal{T}\cup \{t\}$.
  We start with the set system $\mathcal{C}^1\coloneqq\left\{C_u^1~\colon~u \in U\right\}$, where
  $$C_u^1\coloneqq\left\{(s,i) \in \mathcal{T}^+ \times \mathbb{Z}~|~i \leq \dist_{\widehat{G}}(u,s)-\dist_{\widehat{G}}(u,t)\right\}.$$
  As $\widehat{G}$ is planar, by~\Cref{thm:LW} we infer that $\mathcal{C}$ has VC-dimension at most 4.

  We now ``shift columns'' of $\mathcal{C}^1$. That is, define $\mathcal{C}^2 \coloneqq \left\{C_u^2~\colon~u \in U\right\}$, where
  $$C_u^2\coloneqq\left\{(s,i) \in \mathcal{T}^+ \times \mathbb{Z}~|~i \leq \dist_{\widehat{G}}(u,s)+d_{w(s),t}-\dist_{\widehat{G}}(u,t)\right\}.$$
  Clearly, the VC-dimension of $\mathcal{C}^1$ and $\mathcal{C}^2$ are equal: a set $Z \subseteq \mathcal{T}^+ \times \mathbb{Z}$
  shatters $\mathcal{C}^1$ if and only if the set $\{(s,d_{w(s),t}+i)~\colon~(s,i) \in Z\}$ shatters $\mathcal{C}^2$. 

  Now, let $\mathcal{C}^3$ be ``cropped'' $\mathcal{C}^2$:
  $\mathcal{C}^3 \coloneqq \left\{ C_u^3~\colon~u \in U\right\}$, where
  \[ C_u^3 \coloneqq C_u^2 \cap \left(\mathcal{T}^+ \times \{0,\ldots, |R|+1\}\right). \]
  Since restricting to a smaller universe cannot increase VC-dimension, $\mathcal{C}^3$ has VC-dimension at most $4$. 
  Since $|\mathcal{T}^+| = \Oh_k(|R|^2)$, by Sauer-Shelah lemma (Lemma~\ref{lem:sauer-shelah})
  we have $|\mathcal{C}^3| = \Oh_k(|R|^{12})$. 

  Now observe that for every $u_1,u_2 \in U$
  \begin{equation}\label{eq:CtoX}
   C^3_{u_1} = C^3_{u_2} \qquad \mathrm{implies} \qquad \widehat{X}_{u_1} = \widehat{X}_{u_2}. 
  \end{equation}
  Indeed, the assumption $C^3_{u_1} = C^3_{u_2}$ implies that
  for every $w \in A_R(\mathcal{P})$ and $s \in S(W)$ we have
  \begin{align*}
  & \max(0, \min(|R|+1, \dist_{\widehat{G}}(u_1,s)+d_{w,t}-\dist_{\widehat{G}}(u_1,t))) \\ 
  &= \max(0, \min(|R|+1, \dist_{\widehat{G}}(u_2,s)+d_{w,t}-\dist_{\widehat{G}}(u_2,t))).
\end{align*} 
  For fixed $w \in A_R(\mathcal{P})$, we take a minimum of the above expression over all $s \in S(w)$, obtaining:
  \begin{align*}
    & \max(0, \min(|R|+1, \min\{\dist_{\widehat{G}}(u_1,s)~\colon~s \in S(w)\}+d_{w,t}-\dist_{\widehat{G}}(u_1,t))) \\ 
    &= \max(0, \min(|R|+1, \min\{\dist_{\widehat{G}}(u_2,s)~\colon~s \in S(w)\}+d_{w,t}-\dist_{\widehat{G}}(u_2,t))).
  \end{align*} 
  This proves~\eqref{eq:CtoX}.
  From~\eqref{eq:CtoX}, we infer $|\widehat{\mathcal{X}}| \leq |\mathcal{C}^3| = \Oh_k(|R|^{12})$, as desired.
  \end{proof}

  We next relate the distance from a vertex $u\in U$ to $R$ (in $G$) and to $t$ (in $\widehat{G}$).
  \begin{claim}\label{cl:dist}
    For every $u\in U$, $\dist_G(u,R)=\dist_{\widehat{G}}(u,t) - 2|V(G)|$.
  \end{claim}
  \begin{proof}
    Fix $u\in U$.
    We first show that $\dist_G(u,R)\le \dist_{\widehat{G}}(u,t) - 2|V(G)|$.
    For this, consider a shortest path $\widehat{Q}$ in $\widehat{G}$ from $u$ to $t$.
    Observe that there is a vertex $s\in\mathcal{T}$ that is a copy of an anchor vertex $w$,
    such that $\widehat{Q}[s,t]$ is the path from $s$ to $t$ of length $d_{w,t}$
    added in the construction of $\widehat{G}$ from $G\cutgraph H$. Recall that $d_{w,t} =\dist_G(w,R)+|V(G)|$.
    Also, observe that $\widehat{Q}[u,s]$ contains at least one subdivided edge of $E_{\mathsf{out}}$, as it starts
    in $U$ and ends outside $U$, and otherwise corresponds to a walk from $u$ to $w$ in $G$.
    Therefore, we have
    \begin{align*}
      \dist_{\widehat{G}}(u,t)  = |\widehat{Q}| & = |\widehat{Q}[u,s]|+|\widehat{Q}[s,t]|\\ 
        & = |\widehat{Q}[u,s]| + \dist_G(w,R)+|V(G)|\\
        & \ge |V(G)| + \dist_G(u, w) + \dist_G(w,R) + |V(G)|\\
        & \ge \dist_G(u,R) + 2|V(G)|.
    \end{align*}
    
    We next show that $\dist_G(u,R)\ge \dist_{\widehat{G}}(u,t) - 2|V(G)|$.
    For this, consider a shortest path $Q$ in $G$ from $u$ to $R$. Let $y\in R$ be the unique vertex in $R\cap V(Q)$.
    Also, let $z$ be the first vertex of $Q$ (when ordering from $u$ to $y$)
    in $\bigcup_{P\in\mathcal{P}}V(P)$ and let $P\in\mathcal{P}$ be the path that $z$ is contained
    (if $z$ is contained to more than one paths, pick one of them arbitrarily).
    Also, let $w$ be the first vertex of $P[z,x^P]$ (when ordering from $z$ to $x^P$) that is an anchor vertex.
    Observe that $Q[u,z]$ corresponds to a path in $\widehat{G}$ from $u$ to a copy $s'$ of $z$
    that contains exactly one subdivided edge of $E_{\mathsf{out}}$ (and no edge of $E_{\mathsf{next}}$)
    and there is a copy of $P[z,w]$ in $\widehat{G}$ from $s'$ to a copy $s$ of $w$ 
    that contains no edge of $E_{\mathsf{out}} \cup E_{\mathsf{next}}$. 
    Therefore,
    \begin{align*}
      |Q| & = |Q[u,z]|+|Q[z,y]|& \\
      & = |Q[u,z]|+ |P[z,x^P]| & \text{\!\!\!($Q[z,y]$ and $P[z,x^P]$ being shortest paths from $z$ to $R$)}\\
      & = |Q[u,z]|+ |P[z,w]| + |P[w,x^P]|& \\
      & = |Q[u,z]|+ |P[z,w]| + \dist_G(w,R) & \text{($P$ being shortest path from a vertex $v^P$ to $R$)}\\
      & \ge \dist_{\widehat{G}}(u,s) - |V(G)| + d_{w,t} - |V(G)|& \\
      & \ge \dist_{\widehat{G}}(u,t) - 2|V(G)|. &  
    \end{align*} 
    Thus, we have $\dist_G(u,R) = |Q| \ge \dist_{\widehat{G}}(u,t)-2|V(G)|$, as desired.
  \end{proof}

  \begin{claim}\label{cl:dist2}
    For every $u \in U$ and $w \in A_R(\mathcal{P})$, it holds that
    \begin{align*}
    \widehat{\dist}(u,w) < \infty &\quad\mathrm{if\ and\ only\ if} \quad \widehat{\dist}(u,w) = \min\left\{\dist_{\widehat{G}}(u,s)~\colon~s \in S(w)\right\}-|V(G)|,\ \mathrm{and}\\
    \widehat{\dist}(u,w) = \infty &\quad\mathrm{if\ and\ only\ if} \quad \min\left\{\dist_{\widehat{G}}(u,s)~\colon~s \in S(w)\right\} > 2|V(G)|.
    \end{align*}
  \end{claim}
  \begin{proof}
    We first show that if $\widehat{\dist}(u,w) < \infty$, then 
    there exists $s \in S(w)$ with $\dist_{\widehat{G}}(u,s) \leq |V(G)| + \widehat{\dist}(u,w)$. 
    To this end, let $Q$ be a path from $u$ to $w$ in $G$ of length $\widehat{\dist}(u,w)$, as in the definition
    of $\widehat{\dist}(u,w)$. There exists $P \in \mathcal{P}$ with $w \in A_R(P)$ and a vertex $z \in V(P) \cap V(Q)$
    such that $Q$ decomposes into $Q[u,z]$ and $Q[z,w] = P[z,w]$, with all internal vertices of $Q[u,z]$ in $U$. 
    Then, $\widehat{G}$ contains a copy $s'$ of $z$ such that $Q[u,z]$ projects to a path from $u$ to $s'$
    with one subdivided edge of $E_{\mathsf{out}}$ (and no edge of $E_{\mathsf{next}}$) and also a copy of $P[z,w]$ from $s'$
    to a copy $s$ of $w$ with no subdivided edge of $E_{\mathsf{out}} \cup E_{\mathsf{next}}$.
    The concatenation of these two paths witness that $\dist_{\widehat{G}}(u,s) \leq |V(G)| + \widehat{\dist}(u,w)$, as desired.

    To finish the proof of the claim, it suffices to show that if there exists $s \in S(w)$ with
    $\dist_{\widehat{G}}(u,s) \leq 2|V(G)|$, then $\widehat{\dist}(u,w) \leq \dist_{\widehat{G}}(u,s) - |V(G)|$
    (in particular, $\widehat{\dist}(u,w) \neq \infty$).
    To this end, let $\widehat{Q}$ be a path in $\widehat{G}$ from $u$ to $s$ of minimum length. 
    Since $u \in U$ but $s \notin U$, $\widehat{Q}$ necessarily contains at least one subdivided edge of $E_{\mathsf{out}}$. 
    Since $|\widehat{Q}| \leq 2|V(G)|$, $\widehat{Q}$ contains exactly one edge of $E_{\mathsf{out}}$, no edge of 
    $E_{\mathsf{next}}$, and no edge incident with $t$. Consequently, there exists a vertex $s'$ on $\widehat{Q}$
    which is a copy of a vertex $z$ that lies in the prefix of $w$ on some path $P \in \mathcal{P}$ such that 
    $\widehat{Q}$ decomposes as $\widehat{Q}[u,s']$, which has all internal vertices in $U$, and $\widehat{Q}[s',s]$
    going along a copy of $P[z,w]$ to $s \in S(w)$. Hence, $\widehat{Q}$ corresponds to a path $Q$ in $G$
    from $u$ to $w$ that satisfies the requirements of the definition of $\widehat{\dist}(u,w)$
    and witnesses $\widehat{\dist}(u,w) \leq |\widehat{Q}| - |V(G)|$, as desired. 

    This finishes the proof of the claim.
  \end{proof}
  
  Using the two previous claims, we infer that for every $u \in U$ and $w \in A_R(\mathcal{P})$ it holds that
  \begin{equation}\label{eq:prof-to-dist}
    \hatprofile{R}{\mathcal{P}}{u}(w) = \min(|R|+1, \pi[u](w)).
  \end{equation}
  Indeed, 
  \begin{align*}
  \min\left(|R|+1, \pi[u](w)\right) &= \min\left(|R|+1,\min\left\{\dist_{\widehat{G}}(u,s)~\colon~s\in S(w)\right\}+d_{w,t}-\dist_{\widehat{G}}(u,t)\right)\\
  &=\min\big(|R|+1, \min\left\{\dist_{\widehat{G}}(u,s)~\colon~s\in S(w)\right\} - |V(G)| &\\
  &\qquad\qquad\qquad\qquad + \dist_G(w,R)-\dist_G(u,R)\big) &\text{by Claim~\ref{cl:dist}}\\
  &=\min\left(|R|+1, \widehat{\dist}(u,w) + \dist_G(w,R)-\dist_G(u,R)\right)&\text{by Claim~\ref{cl:dist2}}\\
  &=\hatprofile{R}{\mathcal{P}}{u}(w).
  \end{align*}
  Here, in the third step we used the estimate $\dist_G(u,R) - \dist_G(w,R) \leq |U| \leq |V(G)|-|R|$, 
  so if $\min\left\{\dist_{\widehat{G}}(u,s)~\colon~s\in S(w)\right\} > 2|V(G)|$
  (which is equivalent to $\widehat{\dist}(u,w) = \infty$ by Claim~\ref{cl:dist2}),
  then the minimum is attained at value $|R|+1$.

  For every $u\in U$, we set
    $$B_u\coloneqq\left\{(w,i) \in A_R(\mathcal{P}) \times \mathbb{Z}~|~i \leq \hatprofile{G}{R}{u}(w)\right\}.$$
  Claim~\ref{cl:vcd} and~\eqref{eq:prof-to-dist} imply that the set system $\{B_u~\colon~u\in U\}$
    has size $\mathcal{O}_k(|R|^{12})$.

  Now, for every $v\in V(G)$, we set $$S_v\coloneqq\left\{(s,i) \in R \times \{-|R|,\ldots,|R|\}~|~i \leq \distprofile{R}{s_R}{v}(s)\right\}.$$
  The bound on the size of the set system $\{B_u~\colon~u \in U\}$ and~\Cref{lem:hat-to-normal}
  imply that the size of $\left\{S_u~\colon~u \in U \right\}$ is bounded by $\Oh_k(|R|^{12})$.
  We conclude the proof of the lemma by bounding the size of $\left\{S_u~\colon~u \in V(G)\setminus U \right\}$. For this, note that every vertex $v\in V(G)\setminus U$ is either a milestone for some path $P\in\mathcal{P}$ or a vertex in the neutral prefix of a milestone.
  In the latter case, there is a path $P\in\mathcal{P}$ and a milestone $w\in M_R(P)$ such that $S_v=S_w$. Therefore, we have
  $$|\left\{S_u~\colon~u \in V(G)\setminus U \right\}|\le \sum_{P\in\mathcal{P}} |M_R(P)|\le \mathcal{O}_k(|R|^2),$$
  where the second inequality follows from~\eqref{eq:mile}.
  Hence, the size of $\left\{S_v~\colon~v \in V(G)\right\}$ is at most $$|\left\{S_u~\colon~u \in U\right\}|+|\left\{S_u~\colon~u \in V(G)\setminus U\right\}|=\Oh_k(|R|^{12}).$$
  This finishes the proof of Theorem~\ref{thm:distprofiles}.
\end{proof}

\section{Bounded Euler genus graphs with apices: proof of Theorem~\ref{thm:main-apices}}\label{sec:algo-genus}

In this section we prove Theorem~\ref{thm:main-apices}.
(Note that Theorem~\ref{thm:main-genus} is a special case of Theorem~\ref{thm:main-apices}
 for $k=1$.)
We start by deriving the following corollary from \Cref{t:orth_query}.

\begin{corollary}\label{l:max_min_query}
Let $V$ be a set of $n$ points in $\mathbb{R}^d$.
There is a data structure that uses $\Oh \left( d n \log^{d - 2} n \right)$ preprocessing time, $\Oh \left( d n \log^{d - 2} n \right)$ memory and answers the following queries in time $\Oh \left( d \log^{d - 2} n \right)$: given $r_1, \dots, r_d \in \mathbb{R}$, find $\max_{v \in V} \min_{i \in [d]} (v_i + r_i)$, where $v_i$ denotes the $i$th coordinate of $v$.
\end{corollary}

\begin{proof}
    Fix query parameters $r_1, \dots, r_d \in \mathbb{R}$.
    Let $\lambda \coloneqq \max_{v \in V} \min_{i \in [d]} (v_i + r_i)$ denote the answer we want to find.
    
    We say that a pair $(v,i) \in V \times [d]$ is \emph{good} if
    for every $j \in [d]$, it holds that $v_j - v_i \geq r_i -r_j$.
    Let
    $$
    		\lambda' = \max \left\{ v_i + r_i \colon i\in [d], v\in V,\textrm{ and }(v,i)\textrm{ is good} \right\}.
    $$
    We claim that 
    \begin{equation}
    \label{eq:query-lambda}
    \lambda = \lambda'.
   \end{equation}
   Let $v' = \argmax_{v \in V} (\min_{i \in [d]} v_i + r_i)$ and let $i' = \argmin_{i \in [d]} v'_i + r_i$. By the choice of $i'$, for each $j$ we have $v'_j+r_j\geq v'_{i'}+r_{i'}$, implying $v'_j - v'_{i'} \geq r_{i'} - r_j$. Hence $(v',i')$ is good, so $\lambda' \geq v'_{i'} + r_{i'} = \lambda$.
    
    On the other hand, consider a good pair $(v', i')$ maximizing $v'_{i'} + r_{i'}$.
    The goodness of $(v',i')$ implies that $i' = \argmin_{i \in [d]} v'_i + r_i$, hence $\lambda \geq \min_{i \in [d]} v'_i + r_i = v'_{i'} + r_{i'} = \lambda'$. This proves~\eqref{eq:query-lambda}.
    
    For every $i \in [d]$, we set $V_i$ to be the set
    $$ 
    		\{ (v_1 - v_i, v_2 - v_i, \dots, v_{i - 1} - v_i, v_{i + 1} - v_i, \dots, v_d - v_i) : v \in V \}\subseteq \mathbb{R}^{d-1},
    	$$
    	and set $w_i(v) \coloneqq v_i$. Let $\mathbb{D}_i$ be the data structure obtained by applying \Cref{t:orth_query} to $V_i$ and $w_i$. Consider the suffix range
    $$
    		R_i\coloneqq \mathsf{Range}(r_i - r_1, r_i-r_2, \dots, r_i - r_{i - 1}, r_i - r_{i + 1}, \dots,  r_i - r_d)\subseteq \mathbb{R}^{d-1}.
    $$
    Now, by \eqref{eq:query-lambda} we have that
    $$
    		\lambda = \max \left\{ r_i + \max \{ w_i(v) : v \in V_i \cap R_i \}\colon i\in [d] \right\}.
    $$
    This value can be computed by asking $d$ queries to the data structures $\mathbb{D}_i$, for $i\in [d]$. This gives us a data structure satisfying the conditions given in the lemma statement.
\end{proof}

The main work in the proof of Theorem~\ref{thm:main-apices} will be done in the following lemma,
which provides a fast computation of eccentricities once a suitable division is provided on input.
We adopt the notation for divisions introduced in the statement of \Cref{t:r_division}.

\begin{lemma}\label{l:main_ecc}
Fix constants $0 < \alpha, \gamma, \rho < 1$ and $k \in \mathbb{N}$. Assume we are given a connected graph $G$ on $n$ vertices with $O(n)$ edges with positive integer weights, a subset of vertices $X$, a subset of apices $A \subseteq V(G)$ of size at most $k$, and a family $\mathcal{R}$ with $V(G) \setminus A = \bigcup \mathcal{R} $ such that the following conditions are satisfied:
\begin{itemize}[nosep]
	\item $\sum_{R \in \mathcal{R}} |\partial R| \leq \Oh(n^\gamma)$;
	\item for every $R \in \mathcal{R}$, $|R| \leq \Oh(n^\rho)$ and $G[R]$ is connected and contains $O(|R|)$ edges; and
	\item for every $R \in \mathcal{R}$, the number of distance profiles in $G-A$ on $\partial R$ is of $\Oh(n^\alpha)$.
\end{itemize}
Then, we can compute $X$-eccentricity of every vertex of $G$ in time $\Oh(n^{\gamma + 2\rho} \log n + (n^{1 + \gamma} + n^{1 + \alpha}) \log^{k - 1} n)$.
\end{lemma}

\begin{proof}
    Let $G' \coloneqq G - A$ and $X' \coloneqq X \cap V(G')$.    Denote $A \coloneqq \{a_1,a_2,\ldots,a_k\}$. We first describe the procedure, and then discuss its time complexity.

    For~every $a \in A$ and $u \in V(G)$, we compute distance between $a$ and $u$ denoted $d_A(a, u)$.
    
    \medskip
    \emph{Step 1.} \ 
    We start by precomputing the following information for every region $R\in \mathcal{R}$.
    For all $u, v \in R$, we compute the distance between $u$ and $v$ in $G'[R]$, denoted $d_R(u, v)$.
    For all $s \in \partial R, u \in V(G')$, we compute the distance between $s$ and $u$ in $G'$, denoted $d_{\partial R}(u,s)$.
    We arbitrarily pick a pivot vertex $s_R \in \partial R$, and for brevity denote $p_R[u] \coloneqq \distprofile{\partial R}{s_R}{u}$, where the profile is considered in $G'$. That is, $p_R[u]$ is the $(\partial R)$-profile of $u$ with respect to $s_R$:
    $$p_R[u](s) = d_{\partial R}(u, s) - d_{\partial R}(u, s_R),\quad \textrm{for all }u \in V(G')\textrm{ and } s \in \partial R.$$
    We define $P_R \coloneqq \{ p_R[u] \colon u \in V(G')\}$. By our assumption, we have $|P_R| \leq \Oh(n^\alpha)$. Finally, for every profile $p \in P_R$, we list all vertices $v \in X' \setminus R$ such that $p_R[v] = p$ and set up the data structure of \Cref{l:max_min_query} for the points $(d_A(a_1, v), \dots, d_A(a_k, v), d_{\partial R_u}(s_R, v))$; denote it by $\mathbb{D}_{R, p}$.
    
    \medskip
    \emph{Step 2.} \ 
    For every $u \in V(G)$, we compute $\ecc_X(u)$ as follows. If $u \in A$, the answer is $\max_{v \in X} d_A(u, v)$. Hence, we may assume $u \not\in A$.
    Let $R_u$ denote any region of $\mathcal{R}$ containing $u$. For every $v \in R_u$, the shortest path from $u$ to $v$ in $G$ either:
    \begin{itemize}[nosep]
        \item goes through an apex, in which case its length is $\min_{a \in A} d_A(a, u) + d_A(a, v)$; or
        \item is disjoint from $A$ and intersects $\partial R_u$, in which case its length is $\min_{s \in \partial R_u} d_{\partial R_u}(s, u) + d_{\partial R_u}(s, v)$;~or
        \item is contained entirely in $R_u$, in which case its length is $d_{R_u}(u, v)$.
    \end{itemize}
    The length of this path is therefore the minimum among the three quantities.
    Using the above observation, we compute $\dist_G(u, v)$ explicitly for each $v \in R_u$, and define $\Delta^{R_u}_u \coloneqq \max_{v \in R_u \cap X} \dist_G(u, v)$.
    
    For every $v \in V(G) \setminus (A \cup R_u)$, the shortest path between $u$ and $v$ either crosses $A$ or $\partial R_u$. The length of such path avoiding $A$ is
    $$ \min_{s \in \partial R_u} d_{\partial R_u}(s, u) + d_{\partial R_u}(s, v) = 
       d_{\partial R_u}(s_R, v) + \min_{s \in \partial R_u} \left( d_{\partial R_u}(s, u) + p_{R_u}[v](s) \right). $$

We partition the vertices $v$ by their profile $p_{R_u}[v]$ and for every $p \in P_{R_u}$, we compute the maximum distance to a vertex with profile $p$ separately. Let $V_p = \{ v \in X' \setminus R_u \mid p_{R_u}[v] = p \}$. For every $v \in V_p$, we have
    $$ \dist_G(u, v) = \min \left( \min_{a \in A} d_A(a, u) + d_A(a, v), d_{\partial R_u}(s_R, v) + \min_{s \in \partial R_u} \left( d_{\partial R_u}(s, u) + p(s) \right) \right). $$
    We set $r_i \coloneqq d_A(a, u)$ for $i \in [k]$, and $r_{k + 1} \coloneqq \min_{s \in \partial R_u}  \left( d_{\partial R_u}(s, u) + p(s) \right)$. Now,
    $$ \max_{v \in V_p} \dist_G(u, v) = \max_{v \in V_p} \min(r_1 + d_A(a_1, v), \dots, r_k + d_A(a_k, v), r_{k + 1} + d_{\partial R_u}(s_R, v)).$$
    This value can be computed by querying $r_1, \dots, r_{k + 1}$ to the data structure $\mathbb{D}_{R_u, p}$. We define $\Delta^{V(G) \setminus (A \cup R_u)}_u$ as the maximum of such values over all $p \in P_{R_u}$.
    
    Finally, we set $\Delta^{A}_u \coloneqq \max_{a \in A \cap X} d_A(a, u)$, and report $\ecc_X(u) = \max \left(\Delta^{A}_u, \Delta^{R_u}_u, \Delta^{V(G) \setminus (A \cup R_u)}_u \right)$.
    
    It remains to argue that this algorithm can be implemented in the desired running time. For any source $u \in V(G)$, distance from $u$ to all vertices of $G$ can be calculated in time $\Oh((|V(G)| + |E(G)|) \log |V(G)|)$ using Dijkstra's algorithm. Therefore:
    \begin{itemize}[nosep]
        \item computing $d_A(a,\cdot)$ for all $a$ can be done in time $\Oh (n \log n)$,
        \item computing $d_{\partial R}(\cdot,\cdot)$ for all $R$ can be done in time $\Oh (n\log n \cdot \sum_{R \in \mathcal{R}} |\partial R| ) \leq \Oh (n^{1 + \gamma} \log n)$,
        \item computing $d_R(\cdot,\cdot)$ for all $R \in \mathcal{R}$ can be done in time $\Oh (|\mathcal{R}| n^{2\rho} \log n) \leq \Oh (n^{\gamma + 2\rho} \log n)$; constructing $G[R]$ takes $\Oh(|R|^2 \log n) = \Oh(n^{2\rho} \log n)$ time and calculating all pairs shortest paths can be done in time $\Oh(|R||E(G[R])| \log n) = \Oh(n^{2\rho} \log n)$.
    \end{itemize}
    Finally, the total size of the data structures $\mathbb{D}_{R, p}$ over all $R, p$ is $\Oh(|\mathcal{R}| n) = \Oh (n^{1 + \gamma})$, hence we can construct them in time $\Oh (n^{1 + \gamma} \log^{k - 1} n)$.
    
    Consider $u \in V(G) \setminus A$ fixed in step 2. Computing $\Delta^{R_u}_u$ takes $\Oh (|R| \cdot |\partial R_u|)$ time. Computing $\Delta^{A}_u$ can be done in constant time. Computing $\Delta^{V(G) \setminus (A \cup R_u)}_u$ requires asking $|P_{R_u}|$ queries to some $\mathbb{D}_{R, p}$, which takes $\Oh (n^{\alpha} \log^{k - 1} n)$ time in total.
    In total, step 2 for all vertices $u$ can be done in time $\Oh (n^{1 + \alpha} \log^{k - 1} n + n^\rho \cdot \sum_{u \in V(G) \setminus A} |\partial R_u|) = \Oh (n^{1 + \alpha} \log^{k - 1} n + n^{\gamma + 2\rho})$.
    
    We conclude that the total running time is $\Oh(n^{\gamma + 2\rho} \log n + (n^{1 + \gamma} + n^{1 + \alpha}) \log^{k - 1} n)$.
\end{proof}

The next statement is a reformulation of Theorem~\ref{thm:main-apices}. 

\begin{theorem}\label{t:main_bdgenus_apices}
Fix constants $k, g \in \mathbb{N}$. Let $\mathcal{C}$ denote the class of all graphs that can be obtained by taking a graph $G$ of Euler genus bounded by $g$, and adding $k$ apices adjacent arbitrarily to the rest of $G$ and to each other. Then there is an algorithm that given an unweighted graph $G$ belonging to $\mathcal{C}$, together with its set of apices $A$, computes the eccentricity of every vertex in time $\Oh_{k,g} \left( n^{1 + \frac{24}{25}} \log^{k - 1} n \right)$.
\end{theorem}

\begin{proof}
Let $A = \{a_1, \dots, a_k\}$ denote the set of apices and let $G' = G - A$. Fix $\rho\coloneqq \frac{2}{25}$.
Since graphs of bounded genus exclude some fixed clique as a minor,
by \Cref{t:r_division} (with $\varepsilon = \rho/2$)
we can find an $\Oh(n^\rho)$-division $\mathcal{R}$ of $G'$
satisfying $\sum_{R \in \mathcal{R}} |\partial R| = \Oh (n^{1 - \frac{\rho}{2}})$
in time $\Oh(n^{1 + \rho})$ .
By Theorem~\ref{thm:distprofiles}, the graph $G'$ has a degree $12$ polynomial bound on the number of distance profiles. In particular, the number of profiles on every $\partial R$ is of $\Oh(|R|^{12}) = \Oh(n^{12\rho})$. Let $X \coloneqq V(G)$, $\gamma \coloneqq 1 - \frac{\rho}{2} = \frac{24}{25}$ and $\alpha \coloneqq 12\rho = \frac{24}{25}$. Now,  applying \Cref{l:main_ecc} gives us an algorithm computing all eccentricities in time $\Oh(n^{1 + \frac{24}{25}} \log^{k - 1} n)$.
\end{proof}

\section{The general case: Proof of \texorpdfstring{\Cref{thm:main-decomp}}{Theorem 1.6}}\label{sec:algo}

First, we show that data structure of \Cref{l:max_min_query} can be used to compute distances witnessed by shortest paths that pass through a constant-size separator.

\begin{lemma}\label{l:single_adhesion}
Fix a constant $k \in \mathbb{N}$. There exists an algorithm which as the input receives an edge-weighted graph $G$ on $n$ vertices and $m$ edges together with a partition of its vertices into three sets $A, B, C$ such that $|B| \leq k$ and there are no edges between $A$ and $C$, and as the output computes $\max_{c \in C} \dist(a, c)$ for every $a \in A$. The running time is $\Oh(m \log n + n \log^{k - 1} n)$.
\end{lemma}

\begin{proof}
Let $B = \{b_1, \ldots, b_k\}$. For any $a \in A, c \in C$, we have $\dist(a, c) = \min_{i \in [k]} \dist(a, b_i) + \dist(c, b_i)$. First, we run Dijkstra's algorithm from every vertex in $B$ to find $\dist(v, b_i)$ for every $v \in V(G)$ and $i \in [k]$. Next, we use \Cref{l:max_min_query} to construct a data structure $\mathbb{D}$ for the point set $\{(\dist(c, b_1), \dots, \dist(c, b_k))\colon c\in C\}\subseteq \mathbb{R}^k$. Now, the value $\max_{c \in C} \dist(a, c)$ for any given $a$ is equal to the answer of $\mathbb{D}$ to the query with argument $(\dist(a, b_1), \dots, \dist(a, b_k))$.
\end{proof}

After computing the distances over a constant-size separator, we will use the following observation to simplify one of the sides of the separation.

\begin{lemma}\label{l:inserting_paths}
Let $G$ be a edge-weighted connected graph and let $A, B, C$ be a partition of its vertices such that there are no edges between $A$ and $C$. For every pair of vertices $u, v \in B$, let $P_{u, v}$ be any shortest path from $u$ to $v$ with all internal vertices in $C$ (assuming such a path exists).

Let $G'$ denote a graph obtained from $G[A \cup B]$ by adding an edge from $u$ to $v$ of weight equal to the length of $P_{u, v}$, for all $u, v \in B$ for which $P_{u, v}$ exists. Then,  $$\dist_G(s, t) = \dist_{G'}(s, t)\qquad\textrm{for all }s,t\in A\cup B.$$
\end{lemma}
\begin{proof}
Let $G''$ be the graph obtained by adding new edges of $G'$ to $G$.
Fix any $s, t \in A \cup B$ and let $P$ denote the shortest path from $s$ to $t$ in $G''$ which minimizes the number of vertices from $C$ visited. Naturally, the weight of $P$ is equal $\dist_G(s, t)$. Assume that such path visits at least one vertex of $C$. Then, the path $P$ is of the form $s \xrightarrow{P_1} x \xrightarrow{P_2} y \xrightarrow{P_3} t$, where $x, y \in B$ and all the internal vertices of $P_2$ are in $C$. By the construction of $G'$, $P_2$ can be replaced with a direct edge from $x$ to $y$ of the same weight. We obtain a same weight path with a smaller number of vertices of $C$ visited, which is a contradiction. Therefore, $P$ is entirely contained in $A \cup B$, hence it exists in $G'$. This shows that $\dist_G(s, t) = \dist_{G'}(s, t)$.
\end{proof}

The next lemma encapsulates the main algorithmic content of the proof of \Cref{thm:main-decomp}. The algorithm will split the tree decomposition provided on input into smaller parts for which the eccentricities are easier to calculate. We use the following lemma to handle a single such part.
\begin{lemma}\label{l:star}
Fix constants $k, g \in \mathbb{N}, 0 < \delta < \frac{1}{54}$. Assume we are given $n \in \mathbb{N}$, an edge-weighted graph $G$ on at most $n$ vertices with a weight function $w \colon E(G) \to \mathbb{N}$, a vertex subset $A$ and a collection of non-empty vertex subsets $V_0, V_1, \dots, V_\ell$ satisfying the following conditions:
\begin{itemize}[nosep]
	\item The sum of weights of all the edges in $G$ is bounded by $\Oh(n)$.
	\item $V(G) \setminus A = V_0 \cup V_1 \cup \dots \cup V_\ell$.
	\item $|A| \leq k$.
	\item For every $i \in [\ell]$, $G[V_i \setminus V_0]$ is connected, $N_G(V_i \setminus V_0) = V_i \cap V_0$, $|V_i| = \Oh(n^\delta)$, and $|V_0 \cap V_i| \leq 4$.
	\item For all $i, j \in [\ell], i \neq j$, $V_i \setminus V_0$ and $V_j \setminus V_0$ are disjoint and non-adjacent in $G$.
	\item Every edge $uv \in E(G)$ with $u, v \not\in A$ is contained in $G[V_i]$ for some $i\in \{0,1,\ldots,\ell\}$.
	\item The graph obtained by taking $G[V_0]$ and adding a clique on $V_0 \cap V_i$ for every $i \in [\ell]$ has Euler genus bounded by $g$.
\end{itemize}
Then, we can compute the eccentricity of every vertex of $G$ in time $\Oh \left( n^{1 + \frac{150 + 54 \delta}{151}} \log^k n \right)$.
\end{lemma}

\begin{proof}
Fix $\delta' = \frac{1 + 97 \delta}{151}$; we have $\delta' - \delta = \frac{1 - 54\delta}{151} > 0$.
Let $E_i$ denote the set of edges with one endpoint in $V_i$ and the other endpoint in $V_i \setminus V_0$. For $i \in [\ell]$, we shall say that $V_i$ is {\em{heavy}} if the sum of weights of $E_i$ is larger than $n^{\delta'}$. Since the sets $E_i$ are pairwise disjoint and the total sum of weights of all the edges is bounded by $\Oh(n)$, the number of heavy subsets is bounded by $\Oh(n^{1 - \delta'})$. Without loss of generality, we may assume that $V_{\ell' + 1}, \dots, V_\ell$ are heavy and $V_1, \dots, V_{\ell'}$ are not, for some $\ell'\in \{0,\ldots,\ell\}$.

For any source vertex $s$, we can calculate distances from $s$ to every vertex of $G$  using breadth first search in time $\Oh(\sum_{e \in E(G)} w(e)) = \Oh(n)$.
In particular, for every $\ell' < i \leq \ell$, we can compute the distances from every vertex of $V_i$ to every vertex of $G$ in total time $\Oh(n^{2 - \delta' + \delta})$, because $$|V_{\ell'+1}\cup \ldots\cup V_{\ell}|\leq n^{1-\delta'}\cdot \Oh(n^\delta)=\Oh(n^{1-\delta'+
\delta}).$$
Additionally, we calculate distances $\dist_G(a, v)$ for every $a \in A, v \in V(G)$ in time $O(n)$.

For every $i \in [\ell]$ and $u,v \in V_0 \cap V_i$, there exists a shortest path $P_{i,u,v}$ from $u$ to $v$ with all internal vertices belonging to $V_i - V_0$ due to the assumption that $G[V_i - V_0]$ is connected and $N_G(V_i - V_0) = V_i \cap V_0$. Therefore, the distance from $u$ to $v$ is bounded by the sum of weights of edges in $E_i$. In particular, for $i \in [\ell']$, $\dist_G(u, v) \leq n^{\delta'}$.

We define $\widetilde{G}$ to be the graph obtained by taking $G[A \cup V_0 \cup \dots \cup V_{\ell'}]$ and applying the following operation for every $i \in \{\ell' + 1, \dots, \ell\}$:
for each pair of vertices $u, v \in A \cup (V_0 \cap V_i)$, add an edge in $\widetilde{G}$ between $u$ and $v$ with weight equal to the total weight of $P_{i,u,v}$. For a fixed $i, u$, we can find $P_{i, u, v}$ for all $v$ using breadth first search in time $\Oh(n)$. Taking a sum over all $i, u$, we get that $\tilde{G}$ can be computed in total time $\Oh(n^{2 - \delta'})$.

\begin{claim}\label{cl:wG}
The sum of the edge weights in $\widetilde{G}$ is $\Oh(n)$. Moreover, for all $u, v \in V(\widetilde{G})$, we have $\dist_{\widetilde{G}}(u, v) = \dist_{G}(u, v)$.
\end{claim}

\begin{proof}
Consider $i \in \{\ell' + 1, \dots, \ell\}$ and any $u, v \in A \cup (V_0 \cap V_i)$ for which we added an edge. Its weight is bounded by the sum of weights of edges in $E_i$. Therefore, the total weight of all edges added is at most
$$
\sum_{i \in \{\ell' + 1, \dots, \ell\}} \left( |A \cup (V_0 \cap V_i)|^2 \sum_{e \in E_i} w(e) \right) \leq (4 + k)^2 \sum_{e \in E(G)} w(e) = \Oh(n).
$$
This proves the first part of the claim.

For the second part of the claim, consider any $i \in \{\ell' + 1, \dots, \ell \}$ and observe that by our assumptions, $A \cup (V_0 \cap V_i)$ separates $(V_0 \cup \dots \cup V_{\ell'} \cup V_{i + 1} \cup \dots \cup V_\ell) \setminus V_i$ from $V_i \setminus V_0$. Hence it suffices to repeatedly apply \Cref{l:inserting_paths}.
\end{proof}

For every $u \in V(\widetilde{G})$, we have $\ecc_G(u) = \max(\ecc_{\widetilde{G}}(v), \max_{v \in V(G) \setminus V(\widetilde{G})} \dist_G(u, v))$. Note, that we already know all the distances $\dist_G(u, v)$ for $v \in V(G) \setminus V(\widetilde{G})$. Similarly, we can already compute $\ecc_G(u)$ for every $u \in V(G) \setminus V(\widetilde{G})$. Therefore, it remains to compute $\ecc_{\widetilde{G}}(v)$ for each $v \in V(\widetilde{G})$. Our goal is to show that this can be done efficiently using \Cref{l:main_ecc}.

Now, let $G'$ be the graph obtained from $\tilde{G}$ by replacing every edge $e$ non-indicent to $A$ with $w(e)\geq 2$ with a path of length $w(e)$ consisting of unit-weight edges. This operation again preserves the distances. Since the sum of edge weights in $\tilde{G}$ is of $\Oh(n)$, the total number of vertices in $G'$ is of $\Oh(n)$. For $0 \leq i \leq \ell'$, we write $V'_i$ to denote the set $V_i$ together with all the vertices added as a part of a path between two endpoints in $V_i$.
As $V_i$ is not heavy for each $i\in [\ell']$, we have
$$
|V'_i \setminus V'_0| \leq |V_i| + \sum_{e \in E_i} w(e) = \Oh(n^{\delta'})\qquad \textrm{for all }i\in [\ell'].
$$

Let $G_0$ denote the graph $G'[V'_0]$ and let $G_0^*$ denote the graph $G'- A$ with $V'_i - V'_0$ contracted to a single vertex $v_i^*$, for each $i \in [\ell']$; note that, all edges of $G_0$ and $G_0^*$ have unit weight.

\begin{claim}
	The graph $G_0^*$ is does not contain $K_{t}$ as a minor, where $t = \Oh(\sqrt{g})$.
\end{claim}

\begin{proof}
Let $\bar{G}_0$ denote the graph obtained by taking $G_0$ and adding a clique on $V_0 \cap V_i$ for every $i \in [\ell']$.
By lemma assumptions and the fact that subdividing edges does not increase the Euler genus, $\bar{G}_0$ has Euler genus at most $g$. In particular, $\bar{G}_0$ is $K_{t'}$-minor-free for some $t' = \Oh(\sqrt{g})$, because the Euler genus of $K_{t'}$ is $\Omega({t'}^2)$.

Similarly, let $\bar{G}_0^*$ be the graph obtained by taking $G_0^*$ and adding a clique on each $V_0 \cap V_i$.
Note, that $\bar{G}_0^* - \{v_1^*, \dots, v_{\ell'}^*\}$ is precisely $\bar{G}_0$. Let $t = \max(t', 6)$.
Recall that a minor model of a clique $K_t$ consists of $t$ pairwise vertex-disjoint connected subgraphs, called
branch sets, such that there is at least one edge between each pair of the branch sets.
Consider a minor model $\varphi$ of $K_{t}$ in $\bar{G}^*_0$.
Note that $\varphi$ cannot contain any singleton branch set of the form $\{v^*_i\}$, for the degree of $v^*_i$ in $\bar{G}^*_0$ is at most $4 < t - 1$. Furthermore, since $N_{\bar{G}^*_0}(v^*_i) = V_0 \cap V_i$, any branch set containing $v^*_i$ and at least one other vertex contains some $u \in V_0 \cap V_i$, and $N_{\bar{G}^*_0}(v^*_i)\subseteq N_{\bar{G}^*_0}(u)$, hence removing $v^*_i$ from this branch set preserves the model. Therefore, we can assume without loss of generality that all branch sets of $\varphi$ are disjoint from $\{v^*_1, \dots, v^*_{\ell'}\}$, hence $\varphi$ is a minor model of $K_{t}$ in $\bar{G}_0$. This is a contradiction, as $t \geq t'$ and $\bar{G}_0$ is $K_{t'}$-minor-free. Therefore, $\bar{G}_0^*$ is $K_t$-minor-free, hence $G_0^*$ also.
\end{proof}

Let $\rho' = \frac{2 - 108 \delta}{151} > 0$. The graph $G^*_0$ is a unit-weight graph and is $K_{t}$-minor-free.
Hence, by applying \Cref{t:r_division} to $G^*_0$ (with $\varepsilon = \rho'/2$)
we obtain an $n^{\rho'}$-division $\mathcal{R}_0$ in time $\Oh(n^{1 + \rho'})$.
We extend it to $G' - A$ by mapping every contracted vertex $v^*_i$ to $N_{G' - A}[V'_i - V'_0] = (V'_i - V'_0) \cup (V_0 \cap V_i)$. Formally, we put $V''_i \coloneqq N_{G' - A}[V'_i - V'_0]$ and 
$$
\mathcal{R} \coloneqq \left\{ (R_0 \cap V'_0) \cup \bigcup_{i \colon v^*_i \in R_0} V''_i \colon R_0 \in \mathcal{R}_0 \right\}.
$$

Now, we argue that $\mathcal{R}$ is a reasonable division of $G' - A$. Clearly, all sets $R \in \mathcal{R}$ are connected in $G' - A$. Pick any $R \in \mathcal{R}$ and let $R_0$ be its corresponding set in $\mathcal{R}_0$.
Every vertex $v^*_i$ is mapped to a set of size $\Oh(n^{\delta'})$, therefore
$$|R| \leq |R_0| \cdot \Oh(n^{\delta'}) = \Oh(n^{\rho' + \delta'}).$$

By our construction, for every $i \in [\ell']$, $R$ is either disjoint from $V'_i - V'_0$ or contains whole $N_{G' - A}[V'_i - V'_0]$. This means that no vertex belonging to any $V'_i - V'_0$ can be in $\partial R$, hence $\partial R \subseteq V'_0$.

Pick any $u \in \partial R \cap R_0$. Assume that $u \not\in \partial R_0$. Then every vertex of $N_{G_0^*}(u)$ must be in $R_0$, hence $N_{G - A'}(u) \subseteq R$, which is a contradiction. This means that $\partial R \cap R_0 \subseteq \partial R_0$.

Pick any $u \in \partial R - R_0$. Then, $u \in V_0 \cap V_i$ for some $i \in [\ell']$ such that $v_i^* \in R_0$. Moreover, $v_i^* \in \partial R_0$ and is adjacent to $u$ in $G_0^*$. The number of such $u$ is bounded by $4 |\partial R_0 \cap \{ v_1^*, \dots, v_{\ell'}^* \}|$.

Putting two cases together, we obtain:
$$
\sum_{R \in \mathcal{R}} |\partial R| = \sum_{R \in \mathcal{R}} \left(|\partial R \cap R_0| + |\partial R - R_0|\right) \leq \sum_{R_0 \in \mathcal{R}_0} \left(|\partial R_0| + 4 |\partial R_0 \cap \{ v_1^*, \dots, v_{\ell'}^* \}|\right) = \Oh(n^{1 - \frac{1}{2}\rho'}).
$$

It remains to show the following claim.

\begin{claim}
Pick any $R \in \mathcal{R}, s_R \in R$. The number of different distance profiles on $R$ relative to $s_R$ in $G' - A$ is of $\Oh(n^{48\rho' + 54\delta'})$.
\end{claim}
\begin{proof}
We look at every vertex $v \in V(G') \setminus A$ and consider three cases: $v \in R$, $v \in V'_0$, and $v \in V'_i \setminus (V'_0 \cup R)$ for some $i \in [\ell']$. By our construction, $R \cap V'_0$ is non-empty, hence w.l.o.g. we can assume that $s_R \in V'_0$ as whether two vertices have the same profile on $R$ is independent of the choice of the pivot vertex.

In the first case, there are at most $|R| = \Oh(n^{\rho' + \delta'})$ such vertices, hence they realise at most that many profiles.

In the second case, we want to observe that profile of any vertex $v \in V'_0$ on $R$ depends only on its profile on $R \cap V'_0$ (relative to $s_R$). Pick any $t \in R - V'_0$. Then $t \in V'_i - V'_0$ for some $i \in [\ell']$, $V_i \cap V_0 \subseteq R \cap V'_0$, and every path from $v$ to $t$ intersects $V_i \cap V_0$. In particular, distances from $v$ to vertices of $V_i \cap V_0$ determine its distance to $t$, which proves the observation.

Let $\tilde{G}_0$ denote the graph obtained by taking $G'[V'_0]$ and for every $i \in [\ell'], u, v \in V_0 \cap V_i$ adding a disjoint path from $u$ to $v$ of length $\dist(u, v)$. Let $P_i$ denote the vertex set of paths added between $V_0 \cap V_i$. For every $t \in V'_0$ we have $\dist_{G' - A}(v, t) = \dist_{\tilde{G}_0}(v, t)$, so it suffices to bound the number of profiles on $R \cap V'_0$ in $\tilde{G}_0$. By our assumptions, $\tilde{G}_0$ has Euler genus bounded by $g$ and all $P_i$ are of size $\Oh(n^{\delta'})$.

Let $R_0$ be the set of $\mathcal{R}_0$ corresponding to $R$. Let $\tilde{R}_0$ denote the set $(R \cap V'_0) \cup \bigcup_{i : v^*_i \in R_0} P_i$. Such set is connected in $\tilde{G}_0$. Moreover, similarly to $R$, its size is $\Oh(n^{\rho' + \delta'})$. Applying \Cref{thm:distprofiles}, we get that the number of distance profiles on $\tilde{R}_0$ in $\tilde{G}_0$ is $\Oh(n^{12(\rho' + \delta')})$, which also bounds the number of profiles on $R$ in $G' - A$ realised by $V'_0$.

For the third case, assume $v \in V'_i \setminus (V'_0 \cup R)$ for some $i\in [\ell']$. Every path from $v$ to any vertex of $R$ in $G' - A$ intersects $V_i \cap V_0$. Let $v_1, \dots v_p$ be the vertices of $V_i \cap V_0$, where $p \leq 4$. The profile of $v$ on $R$ is then determined by the following:
\begin{itemize}[nosep]
 \item[(a)] the profile of each $v_j$ on $R$,
 \item[(b)] $\dist_{G' - A}(v, v_j) - \dist_{G' - A}(v, v_1)$ for each $2 \leq j \leq p$, and
 \item[(c)] $\dist_{G' - A}(s_R, v_j) - \dist_{G' - A}(s_R, v_1)$ for each $2 \leq j \leq p$ where $s_R$ is some pivot vertex of $R$.
\end{itemize}
By the previous case, the number of distance profiles of each $v_j$ is $\Oh(n^{12(\rho' + \delta')})$. The distances between $v$ and $v_j$ are bounded by $|V'_i|$, hence each quantity described in (b) can take $\Oh(n^{\delta'})$ different possible values. Similarly, since $v_1$ and $v_j$ are connected via $V'_i$, $|\dist_{G' - A}(s_R, v_j) - \dist_{G' - A}(s_R, v_1)| \leq \Oh(n^{\delta'})$. The number of different possible profiles of such $v$ is therefore bounded by $\Oh(n^{48(\rho' + \delta') + 6\delta'}) = \Oh(n^{48\rho' + 54\delta'})$. This finishes the proof of the claim.
\end{proof}

Now we can apply \Cref{l:main_ecc} to graph $G'$ with apex set $A$, $X = V(\widetilde{G})$, and the following constants: $$\rho = \rho' + \delta',\qquad \gamma = 1 - \frac{1}{2}\rho',\quad \textrm{and}\quad \alpha = 48\rho' + 54 \delta'.$$ This allows us to calculate all $V(\widetilde{G})$-eccentricities in $G'$ in time
$$
\Oh \left( \left(
	n^{ 2 - \frac{1}{2} \rho' } +
	n^{ 1 + 48\rho' + 54 \delta' }
\right) \log^k n \right) =
\Oh \left( n^{1 + \frac{150 + 54 \delta}{151}} \log^k n \right).
$$
Since for each $v\in V(\widetilde{G})$ we have $\ecc_{\widetilde{G}}(v) = \max_{u \in V(\widetilde{G})} \dist_{\widetilde{G}}(v, u) = \max_{u \in V(\widetilde{G})} \dist_{G'}(v, u)$, this means that we have successfully computed all the eccentricities in $\widetilde{G}$; and as we argued, this is enough to compute all the eccentricities in $G$ as well.

Finally, the total running time of the algorithm is
$$
\Oh \left( n^{1 + \frac{150 + 54 \delta}{151}} \log^k n + n^{2 - \delta' + \delta} \right) =
\Oh \left( n^{1 + \frac{150 + 54 \delta}{151}} \log^k n \right).
$$\qedhere
\end{proof}

\begin{lemma}\label{l:star2}
Fix constants $k, g \in \mathbb{N}, 0 < \delta < \frac{1}{54}$. Assume we are given $n \in \mathbb{N}$, an edge-weighted graph $G$ on at most $n$ vertices with a weight function $w \colon E(G) \to \mathbb{N}$, a vertex subset $A$ and a collection of non-empty vertex subsets $V_0, V_1, \dots, V_\ell$ satisfying the same conditions as in \Cref{l:star} with the following differences:
\begin{itemize}
	\item we don't require $G[V_i - V_0]$ to be connected and $V_i - V_0$ to be adjacent to whole $V_i \cap V_0$;
	\item instead of $|V_0 \cap V_i| \leq 4$, we require $|V_0 \cap V_i| \leq k$.
\end{itemize}
Then, we can compute the eccentricity of every vertex of $G$ in time $\Oh \left( n^{1 + \frac{150 + 54 \delta}{151}} \log^{k + 5g} n \right)$.
\end{lemma}

\begin{proof}
We will reduce our input to one which will satisfy the conditions of \Cref{l:star}. We start by addressing the adhesions $V_0 \cap V_i$ containing too many vertices.

Let $G_0$ denote the graph $G[V_0]$ with cliques placed at $V_0 \cap V_i$ for every $i \in [\ell]$.
For every $i \in [\ell]$ we repeat the following procedure: while $|V_0 \cap V_i| > 4$,
remove arbitrary $5$ vertices from $V_0 \cap V_i$. Since $|V_0 \cap V_i| \leq k$ for each $i\in [\ell]$,
this procedure can be implemented in total time $\Oh(n)$. As a result, at the end we have $|V_0 \cap V_i| \leq 4$ for all $i \in [\ell]$. Let $M$ be the set of all the removed vertices. By our assumptions, $G_0$ has Euler genus bounded by $g$, hence it cannot contain $g + 1$ pairwise disjoint copies of $K_5$
(as the Euler genus of a graph is the sum of the Euler genera of its 2-connected components~\cite{StahlB77} and $K_5$ is not planar). Each removed quintiple of vertices induces a $K_5$ in $G_0$, hence we have $|M| \leq 5g$. We set $A' = A \cup M$ and may thus assume that $V_i$ is disjoint from $A'$ for all $0 \leq i \leq \ell$.

Now, fix $i \in [\ell]$. Let $C^i_1, \dots, C^i_{r_i}$ denote the connected components of $V_i - V_0$ in $G - A'$. We define $W^i_j := N_{G - A'}[C^i_j]$ for every $j \in [r_i]$. Clearly, all $W^i_j$ induce a connected subgraph of $G$ and satisfy $N_{G - A'}(W^i_j - V_0) = W^i_j \cap V_0$. We put $V'_0 := V_0$ and enumerate
$$
\{V'_1, V'_2, \dots V'_{\ell'}\} := \{ W^i_j \colon i \in [\ell], j \in [r_i] \}.
$$
It is easy to verify that the sets $A'$ and $V'_0, V'_1, \dots, V'_{\ell'}$ satisfy the conditions of \Cref{l:star}. We apply said lemma to calculate the eccentricity of every vertex of $G$ in the desired time.
\end{proof}

The next statement is a reformulation of \Cref{thm:main-decomp}.

\begin{theorem}
Fix constants $k, g \in \mathbb{N}$. Assume we are given a graph $G$ on $n$ vertices together with its tree decomposition $(T, \beta)$ and a set of private apices $A_t \subseteq \beta(t)$ for each node $t\in V(T)$ such that the following conditions hold:
\begin{itemize}[nosep]
 \item For every node $t \in V(T)$, we have $|A_t| \leq k$.
 \item For every edge $st \in E(T)$,  we have $|\beta(v) \cap \beta(u)|\leq k$.
 \item For every node $t \in V(T)$, graph obtained by taking $G[\beta(t)] - A_t$ and turning  $(\beta(t) \cap \beta(s))\setminus A_t$ into a clique for every edge $st \in E(T)$ has Euler genus bounded by $g$.
\end{itemize}
Then, we can compute the eccentricity of every vertex of $G$ in time $\Oh \left( n^{1 + \frac{355}{356}} \log^{k + 5g} n \right)$.
\end{theorem}

\begin{proof}
We may assume that $|V(T)|\leq n$, for every tree decomposition with no two bags comparable by inclusion has this property; and adjacent comparable bags can be merged by contracting the edge between them.

For a node $t\in V(T)$, by the {\em{weight}} of $t$ we mean the size of the corresponding bag, that is, $|\beta(t)|$. For any subset of nodes $S \subseteq V(T)$, we define $\beta(S) \coloneqq \bigcup_{t \in S} \beta(t)$ By the {\em{weight}} of $S$, we mean the total weight of the elements of $S$, that is, $\sum_{t\in S} |\beta(t)|$. 

\begin{claim}\label{cl:weight-T}
The weight of $V(T)$ is of $\Oh(n)$.
\end{claim}

\begin{proof}
The sets $\beta'(t) := \beta(t) - \bigcup_{s \in N_T(t)} \beta(s)$ are pairwise disjoint. We have
$$
\sum_{t \in V(T)} |\beta(t)| =
\sum_{t \in V(T)} |\beta'(t)| + 2 \cdot \sum_{st \in E(T)} |\beta(s) \cap \beta(t)| \leq
|V(T)| + 2k|E(T)| = \Oh(n).
$$
\end{proof}

Since every bag induces a graph of bounded Euler genus, the number of edges contained in a bag is linear in its size. In particular, this implies that the total number of edges of $G$ is also bounded by $\Oh(n)$.

We set $$\delta \coloneqq \frac{1}{356}\qquad\textrm{and}\qquad \Delta \coloneqq \frac{355}{356}.$$ Root the tree $T$ in an arbitrarily chosen node; this naturally imposes an ancestor-descendant relation in $T$ (for convenience, every node is considered its own ancestor and descendant).

We start by partitioning $T$ into connected subtrees using the following procedure.
We proceed bottom-up over $T$, processing nodes in any order so that a node is processed after all its strict descendants have been processed. Along the way, we mark some nodes and split the edges of $T$ into heavy and light. Let $t \in V(T)$ be the currently processed non-root node of $T$ and let $e \in E(T)$ be the edge connecting $t$ with its parent. If the total weight of all the unmarked nodes that are descendants of $t$ is at least $n^\delta$ (recall that this includes $t$ itself as well), then we declare $e$ heavy and mark all the descendants of $t$ that were unmarked so far. Otherwise, the edge $e$ is declared light and the procedure proceeds to further nodes of $T$.

Observe that
removing all heavy edges splits $T$ into connected subtrees, say $T'_1, \cdots T'_m$. All of the subtrees, except for possibly the subtree containing the root node, are of weight at least $n^\delta$. In particular, the number of subtrees $m$, and therefore the number of heavy edges, is  bounded by $\Oh(n^{1 - \delta})$. Moreover, in every subtree $T'_i$, removing the node closest to the root splits $T'_i$ into smaller components, each of weight less than $n^\delta$.

Fix a heavy edge $e$ and let $T^e_1$ and $T^e_2$ be the two subtrees into which $T$ splits after removing~$e$. Let $X^e_i = \beta(T^e_i)$ for $i \in \{1, 2\}$. Put $A_e = X^e_1 \setminus X^e_2$, $C_e = X^e_2 \setminus X^e_1$, and $B_e = X^e_1 \cap X^e_2$. By the properties of tree decompositions, such choice of $A_e, B_e, C_e$ satisfies the conditions of \Cref{l:single_adhesion}, hence in time $\Oh(n \log^{k - 1} n)$ we can compute $\max_{v \in X^e_2} \dist_G(u,v)$ for every $u \in X^e_1$, and $\max_{u \in X^e_1} \dist_G(u,v)$ for every $v \in X^e_2$. Computing this for every heavy edge $e$ takes total time $\Oh(n^{2 - \delta} \log^{k - 1} n)$.

Fix any subtree $T'=T'_j$. Let $e_1 = t^{e_1}_1t^{e_1}_2, e_2 = t^{e_2}_1 t^{e_2}_2, \dots, e_\ell = t^{e_\ell}_1 t^{e_\ell}_2$ denote the heavy edges incident to $T'$, where $t^{e_i}_1 \in V(T')$ and $V(T') \subseteq V(T_1^{e_i})$ for every $i \in [\ell]$.
For a vertex $v \in \beta(T')$, let
$$d_0(v) = \max_{u \in \beta(T')} \dist_G(v, u)\qquad\textrm{and}\qquad d_i(v) = \max_{u \in X_2^{e_i}}\dist_G(v,u),\quad\textrm{for } i \in [\ell].$$ We have $\ecc(v) = \max \{ d_i(v)\colon i\in \{0,1,\ldots,\ell\}\}$.The values of $d_i(v)$ are already calculated for all $i\in [\ell]$, hence it remains to compute $d_0(v)$.

For every $i \in [\ell]$ and every pair of vertices $u, v \in \beta(t^{e_i}_1) \cap \beta(t^{e_i}_2)$ we find a shortest path between $u$ and $v$ with all internal vertices inside $X^{e_i}_2$ (or determine that it doesn't exist). For a fixed $u, v$ this can be done in time $\Oh(n)$. Since in total we perform this step at most $2k^2$ times per heavy edge, it takes $\Oh(n^{2 - \delta})$ time in total. Let $P_{i, u, v}$ denote such path, assuming it exists.

Let $G'$ denote the graph obtained from $G[\beta(T')]$ by taking every $i, u, v$ for which $P_{i, u, v}$ exists and adding an edge between $u$ and $v$ of weight equal to the total weight of $P_{i, u, v}$.
The weight of every edge inserted in $\beta(t^{e_i}_1) \cap \beta(t^{e_i}_2)$ is bounded by $|X^{e_i}_2|+1$. The total weight of all edges inserted is therefore at most
$$
\sum_{i \in [\ell]} |\beta(t^{e_i}_1) \cap \beta(t^{e_i}_2)|^2 \cdot (|X^{e_i}_2|+1) \leq
k^2 \sum_{i \in [\ell]} (|X^{e_i}_2|+1) = \Oh(n),
$$
where the last equality follows from the fact that all the trees $T^{e_i}_2$ are pairwise disjoint.
By \Cref{l:inserting_paths}, we have $\dist_{G'}(u, v) = \dist_G(u, v)$ for each $u, v \in \beta(T')$. Hence, computing $d_0(v)$ for every $v \in \beta(T')$ is equivalent to computing the eccentricity of every vertex in $G'$.

If the size of $\beta(T')$ is smaller than $n^\Delta$, we compute the eccentricities naively in time $\Oh(|\beta(T')|^2)$, 
noting that $G'$ has $\Oh(|\beta(T')|)$ edges (thanks to Claim~\ref{cl:weight-T} and bounded genus assumption 
of the last bullet of the theorem statement). Otherwise, we argue that we can use the algorithm in \Cref{l:star} as follows.

Let $t$ be the node of $T'$ closest to the root. Let $s_1, \dots, s_p$ be the children of $t$ in $T$ and let $T''_i$ denote the connected component of $T' - \{t\}$ containing $s_i$. Set $V_0 = \beta(t)$ and $V_i = \beta(T''_i)$ for $i \in [p]$.

It is now easy to verify that $G'$ and sets $A, \{V_i\colon 0\leq i\leq p\}$ selected this way satisfy the assumptions of \Cref{l:star2}. This allows us to use it to compute the eccentricities in $G'$ in time
$$
\Oh \left( n^{1 + \frac{150 + 54\delta}{151}} \log^{k + 5g} n \right) =
\Oh \left( n^{1 + \frac{354}{356}} \log^{k + 5g} n \right).
$$
As we argued, from these eccentricities, we may easily compute all the eccentricities in $G$.

Now, let us analyse the total running time of the whole algorithm. We invoke \Cref{l:star} $\Oh(n^{1 - \Delta})$ times, since we apply it only to subtrees $T'_i$ of size at least $n^\Delta$. The total running time of those applications is hence
$$
\Oh \left( n^{2 + \frac{354}{356} - \Delta} \log^{k + 5g} n \right) =
\Oh \left( n^{1 + \frac{355}{356}} \log^{k + 5g} n \right).
$$
We compute the eccentricities naively for subtrees smaller than $n^\Delta$, hence the total running time of this computation is
$$
\sum_{i \in [m] \colon |\beta(T'_i)| \leq n^\Delta} |\beta(T'_i)|^2 \leq
n^\Delta \cdot \sum_{i \in m} |\beta(T'_i)| = \Oh(n^{1 + \Delta})=\Oh\left(n^{1+\frac{355}{356}}\right).
$$
The rest of computation can be done in $\Oh(n^{2 - \delta} \log^k n)$. Therefore, the whole algorithm runs in time $\Oh \left( n^{1 + \frac{355}{356}} \log^{k + 5g} n \right)$.
\end{proof}

\section*{Acknowledgements}
Marcin thanks Jacob Holm, Eva Rotenberg, and Erik Jan van Leeuwen
for many discussions on subquadratic algorithms for diameter while his stay on sabbatical
in Copenhagen.

\bibliographystyle{plain}

\end{document}